%% file: main.tex
\documentclass{llncs}

\usepackage[utf8]{inputenc}
\usepackage{url}
\usepackage{amssymb}
\usepackage{amsmath}
\usepackage{graphicx}
\usepackage{xspace}
\usepackage{color}
\usepackage{multirow}
\usepackage{tikz}
\usetikzlibrary{patterns,positioning,calc}
\usepackage{proof}
\usepackage{prooftree}
\usepackage{thm-restate}
\usepackage{multirow}
\usepackage{float}
\usepackage{wrapfig}

\input{header_comm}             

\newcommand{\I}{{\mathrel{\mathcal{I}}}}
\newcommand{\ie}{\emph{i.e.,}\xspace}
\newcommand{\eg}{\emph{e.g.,}\xspace}

\newcommand{\tr}{\mathsf{tr}}
\newcommand{\approxc}{\approx_c}

\newcommand{\Ch}{\mathcal{C}}
\newcommand{\proc}[2]{(#1;#2)}

\newcommand{\eqdef}{\stackrel{\mathsf{def}}{=}}
\newcommand{\refer}{\triangleright}

\newcommand{\X}{\mathcal{X}}
\newcommand{\W}{\mathcal{W}}
\newcommand{\T}{\mathcal{T}}
\newcommand{\N}{\mathcal{N}}

\newcommand{\triple}[3]{(#1;#2;#3)}
\newcommand{\q}{\mathcal{Q}}
\newcommand{\p}{\mathcal{P}}
\newcommand{\E}{\mathsf{E}}
\newcommand{\dom}{\mathrm{dom}}

\newcommand{\statequiv}{\sim}

\makeatletter
\def\rightarrowfillstar@{\arrowfill@\relbar\relbar{\rightarrow\smash{^*}}}
\newcommand{\xrightarrowstar}[2][]{\ext@arrow
  0{13}{15}8\rightarrowfillstar@{#1}{#2}}
\newcommand{\lrstep}{\@ifstar{\xrightarrowstar}{\xrightarrow}}
\makeatother

\newcommand{\cb}{c_B}
\newcommand{\bio}{\mathtt{io}}
\makeatletter
\def\Rightarrowfillstar@{\arrowfill@=={\Rightarrow\smash{^*}}}
\newcommand{\NewxRightarrowstar}[2][]{\ext@arrow
  0{13}{15}8\Rightarrowfillstar@{#1}{#2}}
\newcommand{\NewxRightarrow}[2][]{\ext@arrow
  0{13}{15}8\Rightarrowfill@{#1}{#2}}
\newcommand{\LRstep}{\@ifstar{\NewxRightarrowstar}{\NewxRightarrow}}
\makeatother

\title{A reduced semantics for deciding trace equivalence using
  constraint systems
\thanks{This work has been
    partially supported by the  project JCJC VIP  ANR-11-JS02-006, and
    the Inria large scale initiative CAPPRIS.}
}

\author{David Baelde\inst{1} \and St\'ephanie Delaune\inst{1} \and Lucca Hirschi\inst{1,2}}

\institute{
LSV, ENS Cachan \& CNRS \& Inria Saclay \^Ile-de-France
\and
ENS Lyon, France
}

\begin{document}

\maketitle
\begin{abstract}
  Many privacy-type properties of security protocols can be modelled
   using trace equivalence properties in suitable process algebras. 
  It has been shown that such properties can be decided for interesting
  classes of finite processes (\ie without replication) by means of  symbolic execution
  and constraint solving.
  However, this does not suffice to obtain practical tools. Current
  prototypes suffer from a classical combinatorial explosion problem caused
  by the exploration of many interleavings in the behaviour of processes.
  M\"odersheim \emph{et al.}~\cite{ModersheimVB10} have tackled this problem
  for reachability properties using partial order reduction techniques. We revisit their work, generalize it and
  adapt it for equivalence checking. 
We obtain an optimization in the form of a reduced symbolic semantics
that eliminates redundant interleavings on the fly.
\end{abstract}

\input{intro}

\input{model}

\input{compression}

\input{symbolic}

\input{differentiation}

\input{conclu}


\bibliographystyle{abbrv}
\bibliography{reference}

\newpage
\appendix

\input{app-compression}
\input{app-symbolique}

\input{app-differentiation}

\end{document}

%% file: header_comm.tex
\usepackage{centernot}
\usepackage{mathtools}
\usepackage{stmaryrd}

\usepackage{placeins}

\makeatletter
\newcommand{\rightdoublearrow}{%
  \rightarrow\mkern-10mu\protect\joinrel\rightarrow}

\newcommand{\rightdoublearrowfill@}
  {\arrowfill@\relbar\relbar\rightdoublearrow}
\newcommand{\mapstodoublearrowfill@}
  {\arrowfill@{\mapstochar\relbar}\relbar\rightdoublearrow}
\newcommand{\xrightdoublearrow}[2][]
  {\ext@arrow 3{15}59\rightdoublearrowfill@{#1}{#2}}
\newcommand{\xmapstodoublearrow}[2][]
  {\ext@arrow 3{15}59\mapstodoublearrowfill@{#1}{#2}}
\makeatother

\newcommand{\Bio}{\mathcal{B}}

\newcommand{\Ind}{\mathrel{\mathcal{I}}}
\newcommand{\act}{\mathsf{act}}
\newcommand{\mini}{\mathsf{min}}

\newcommand{\init}{\mathsf{init}}

\newcommand{\C}{\mathcal{C}}
\newcommand{\cs}[2]{(#1;#2)}
\newcommand{\Set}{\mathcal{S}}
\newcommand{\eqi}{\mathop{{=}^{?}}}
\newcommand{\neqi}{\mathop{{\neq}^{?}}}
\newcommand{\dedi}[3]{#1 \mathop{\vdash^{?}_{#2}} #3} 
\newcommand{\Sol}{\mathsf{Sol}}

\makeatletter
\newcommand{\xMapsto}[2][]{\ext@arrow 0599{\Mapstofill@}{#1}{#2}}
\def\Mapstofill@{\arrowfill@{\Mapstochar\Relbar}\Relbar\Rightarrow}
\makeatother


\newcommand{\pair}[2]{\langle #1,#2 \rangle}
\newcommand{\test}[3]{\mathtt{if}\ #1\ \mathtt{then}\ #2\ \mathtt{else}\ #3}
\newcommand{\testt}[2]{\mathtt{if}\ #1\ \mathtt{then}\ #2}
\newcommand{\Out}{\mathtt{out}}
\newcommand{\In} {\mathtt{in}}
\newcommand{\ok}{\mathsf{ok}}
\newcommand{\ska}{\mathit{ska}}
\newcommand{\skb}{\mathit{skb}}
\newcommand{\start}{\mathsf{start}}
\newcommand{\aenc}[2]{\mathsf{aenc}(#1,#2)}
\newcommand{\adec}[2]{\mathsf{adec}(#1,#2)}
\newcommand{\pk}[1]{\mathsf{pk}(#1)}
\newcommand{\pub}{\mathsf{pub}}

\newcommand{\projl}[1]{\pi_1(#1)}
\newcommand{\projr}[1]{\pi_2(#1)}

\newcommand{\fv}{\mathit{fv}} 
\newcommand{\fvs}{\mathit{fv}^2} 
\newcommand{\fvp}{\mathit{fv}^1} 
\newcommand{\inp}[2]{\In(#1,#2)} 
\newcommand{\out}[2]{\Out(#1,#2)} 

\newcommand{\sint}[1]{\lrstep{#1}} 
\newcommand{\sintf}[2]{\xrightdoublearrow{#1}_{#2}} 
\newcommand{\sintc}[1]{\lrstep{#1}_c}         
\newcommand{\ssym}[1]{\xmapsto{#1}}           
\newcommand{\ssymf}[2]{\xmapstodoublearrow{#1}_{#2}}   
\newcommand{\ssymc}[1]{\xmapsto{#1}_c}        
\newcommand{\ssymd}[1]{\xmapsto{#1}_d}          

\newcommand{\eint}{\approx}                      
\newcommand{\esymd}{\approx_d^{2}}                   
\newcommand{\esymdf}{\approx_d^{1}}                   
\newcommand{\incesymd}{\sqsubseteq_d^{2}} 
\newcommand{\incesymdf}{\sqsubseteq_d^{1}} 
\newcommand{\incesym}{\sqsubseteq_s} 
\newcommand{\inceintc}{\sqsubseteq_c} 
\newcommand{\esymdp}{\approx_d^{1}}                   
\newcommand{\incesymdp}{\sqsubseteq_d^1}
\newcommand{\esym}{\approx_s}                      

\newcommand{\estat}{\sim}                        
\newcommand{\och}{\prec}                   
\newcommand{\vect}[1]{\overrightarrow{#1}} 
\newcommand{\io}[3]{\mathtt{io}_{\mathit{#1}}(\vect{#2},\vect{#3})}
\newcommand{\iossvect}[3]{\mathtt{io}_{\mathrm{#1}}({#2},{#3})}

\newcommand{\Dep}[2]{\mathrm{dep}\left(#1,#2\right)} 
\newcommand{\AllDep}[1]{\mathrm{Deps}\left(#1\right)}
\newcommand{\depSym}{\mathbb{n}}                     
\newcommand{\constrd}[2]{#1 \depSym #2}       
\newcommand{\noconstrd}[1]{{#1}^\circ} 

\makeatletter
\renewcommand*\env@matrix[1][\arraystretch]{%
  \edef\arraystretch{#1}%
  \hskip -\arraycolsep
  \let\@ifnextchar\new@ifnextchar
  \array{*\c@MaxMatrixCols c}}
\makeatother

%% file: intro.tex
\section{Introduction}
\label{sec:intro}

Security protocols are widely used today to secure transactions that rely on 
public channels like the Internet, where dishonest users may listen to 
communications and interfere with them. 
A secure communication has a different meaning depending on the underlying 
application. It ranges from the confidentiality of data (medical files, 
secret keys, etc.) to, \eg verifiability in electronic voting systems. Another 
example is the notion of privacy that appears in many contexts such as 
vote-privacy in electronic voting or untraceability in RFID technologies.

Formal methods have proved their usefulness for precisely analyzing the 
security of protocols. 
In particular, a wide variety of model-checking approaches have been 
developed
to analyse protocols against an 
attacker who entirely controls the communication network, and
several tools are now available to automatically verify cryptographic 
protocols~\cite{BlanchetCSFW01,Cr2008Scyther,sysdesc-CAV05}.
A major challenge faced here is that one has to account for infinitely
many behaviours of the attacker, who can generate arbitrary messages.
In order to cope with this prolific attacker problem and obtain decision
procedures, approaches based on symbolic semantics and constraint resolution
have been proposed~\cite{MS02,RT01}.
This has lead to tools for verifying reachability-based security properties 
such as confidentiality~\cite{MS02} or, more recently, equivalence-based
properties such as privacy~\cite{Tiu-csf10,cheval-ccs2011,CCK-esop12}.

In both cases, the practical impact of most of these tools is limited by a typical
state explosion problem caused by the exploration of the large number of
interleavings in the protocol's behaviour.
In standard model-checking approaches for concurrent systems, the 
interleaving problem is handled using partial order reduction 
techniques~\cite{Peled98}. 
For instance, the order of execution of two 
independent (parallel) actions is typically irrelevant for checking
reachability.
Things become more complex when working with a symbolic semantics:
the states obtained from the interleaving of parallel actions will differ,
but the sets of concrete states that they represent will have a significant
overlap.
Earlier work has shown how to limit this overlap~\cite{ModersheimVB10}
in the context of reachability properties for security protocols,
leading to high efficiency gains in the OFMC tool of the
AVISPA platform~\cite{sysdesc-CAV05}.

In this paper, we revisit the work of~\cite{ModersheimVB10} to obtain a 
partial order reduction technique for the verification of equivalence 
properties.
Specifically, we focus on trace equivalence, requiring that two processes
have the same sets of observable traces and perform indistinguishable
sequences of outputs.
This notion is well-studied and several algorithms and tools support
it~\cite{BlanchetAbadiFournetJLAP08,chevalier10,Tiu-csf10,cheval-ccs2011,CCK-esop12}.
Contrary to what happens for reachability-based properties, trace equivalence 
cannot be decided relying only on the reachable states. The sequence of 
actions that leads to this state plays a role. Hence, extra precautions
have to be taken before discarding a particular interleaving: we have to 
ensure that this is done in both sides of the equivalence in a similar
fashion.
Our main contribution is an optimized form of equivalence that discards
a lot of interleavings, and
a proof that this reduced equivalence coincides with
trace equivalence.  Furthermore, our study brings
an improvement of the original technique~\cite{ModersheimVB10} that
would apply equally well for reachability checking.

\smallskip{}

\noindent\emph{Outline.}
In Section~\ref{sec:model}, we introduce our model for security processes.
We consider the class of simple processes introduced in~\cite{CCD-tcs13},
with else branches and no replication.
Then we present two successive optimizations in the form of refined
semantics and associated trace equivalences.
Section~\ref{sec:compression} presents 
a \emph{compressed} semantics that limits interleavings by
executing blocks of actions. 
Then, this is lifted to a symbolic semantics in 
Section~\ref{sec:constraint-solving}.
Finally, Section~\ref{sec:diff} presents the \emph{reduced} semantics
which makes use of dependency constraints to remove more interleavings.
We conclude in Section~\ref{sec:conclu}, mentioning a preliminary implementation
that shows efficiency gains in practice and 
some directions for future work.

%% file: model.tex
\section{Model for security protocols}
\label{sec:model}

In this section, we introduce the cryptographic process calculus that
we will use to describe security protocols. This calculus is close to
the applied pi calculus~\cite{AbadiFournet2001}.

\subsection{Messages}
\label{subsec:messages}

A protocol consists of some agents communicating on a network.
Messages sent by agents are modeled using a term algebra.
We assume two infinite and disjoint sets of variables, $\X$ and $\W$.
Members of $\X$ are denoted $x$, $y$, $z$, whereas members of~$\W$ are
denoted~$w$ and used 
as \emph{handles} for previously
output terms.
We also assume a set~$\N$ of \emph{names}, which are used
for representing keys or nonces,
and a signature $\Sigma$ consisting of a finite set of function symbols.
Terms are generated inductively from names, variables, and function 
symbols applied to other terms. For $S \subseteq \X\cup\W\cup\N$,
the set of terms built from $S$ by applying function symbols in
$\Sigma$ is denoted by $\T(S)$. 
Terms in $\T(\N\cup\X)$ represent messages and are denoted by
$u$, $v$, etc. while terms in $\T(\W)$ represent \emph{recipes}
(describing how the attacker built a term from the available outputs)
and are written $M$, $N$, $R$. 
We write $\fv(t)$ for the set of variables (from $\X$ or $\W$)
occurring in a term $t$. A term is \emph{ground} if it does not contain any
variable, \ie it belongs to $\T(\N)$. We may rely on a sort system for terms, but its details
are unimportant for this paper.

To model algebraic properties of cryptographic primitives, we consider
an equational theory $\E$. The theory will usually be generated for finite
axioms and enjoy nice properties, but these aspects are irrelevant for the
present work.

\begin{example}
\label{ex:signature}
In order to model asymmetric encryption and pairing, we consider:\\[1mm]
\null\hfill $\Sigma = \{ \aenc{\cdot}{\cdot},\; \adec{\cdot}{\cdot}, \;\pk{\cdot}, \;\pair{\cdot}{\cdot}, \;\projl{\cdot},\; \projr{\cdot}\}.$\hfill\null

\smallskip{}

To take into account the properties of
these operators, we consider the equational
theory~$\E_{\mathsf{aenc}}$ generated by the three following equations:\\[1mm]
\null\hfill
$\adec{\aenc{x}{\pk{y}}}{y} =x, \;\;\;\;
\projl{\pair{x_1}{x_2}} = x_1, \mbox{ and }
\projr{\pair{x_1}{x_2}} = x_2. 
$\hfill\null

\smallskip{}

\noindent For instance, we have $\projr{\adec{\aenc{\pair{n}{\pk{ska}}}{\pk{skb}}}{skb}}  =_{\E_\mathsf{aenc}} \pk{ska}$.
\end{example}

\subsection{Processes}
\label{subsec:processes}

We do not need the full applied pi calculus to represent security
protocols.  Here, 
we only consider public channels and we assume that each process communicates on
a dedicated channel.


Formally, we assume a set $\Ch$ of \emph{channels} and we consider the fragment of
\emph{simple processes} without replication built on \emph{basic processes} as defined in~\cite{CCD-tcs13}.
A basic process represents a party in a protocol, which may sequentially
perform  actions such as waiting for a message, checking that a message has a 
certain form, or outputting a message.
Then, a simple process is a parallel composition of such basic processes
playing on distinct channels.

\begin{definition}[basic/simple process]
\label{def : basic process}
The set  of \emph{basic processes} 
on ${c \in \Ch}$
is defined using the following grammar (below $u,v\in\T(\N\cup\X)$ and $x \in \X$):\\[1mm]
\null\hfill
$\begin{array}{lcll}
  P,Q &:=& 0                            &\mbox{null} \\
      &|& \test{u=v}{P}{Q}  \quad\quad &\mbox{conditional}\\
      &|& \In(c,x).P                   &\mbox{input}\\
      &|& \Out(c,u).P                  &\mbox{output}\\
\end{array}$
\hfill\null

\smallskip{}

 A \emph{simple process} $\p = \{ P_1, \ldots, P_n \}$
  is a {multiset} of basic processes $P_i$ on pairwise
  distinct channels $c_i$. {We assume that null processes are removed.}
\end{definition}

For conciseness, we often omit brackets, null processes, and even
``\texttt{else} 0''.
Basic processes are denoted by the letters $P$ and $Q$, whereas simple
processes are denoted using $\p$ and $\q$.

\smallskip{}

During an execution, the attacker learns the messages that have been
sent on the different public channels. Those messages are organized
into a \emph{frame}. 

\begin{definition}[frame]
  A \emph{frame} $\Phi$ is a substitution whose domain is included in
  $\W$ and image is included in $\T(\N\cup\X)$. It is written
  $\{ w \refer u, \ldots \}$. A frame is \emph{closed} when
  its image only contains ground terms.
\end{definition}

An \emph{extended simple proces} (denoted $A$ or $B$) is a pair made of a simple process
and a frame.  Similarly, we define \emph{extended basic processes}.
Note that we
 do not have an explicit set of restricted names. Actually, all
 names are restricted and public ones are explicitly given to
 the attacker through a frame.

\begin{example}
\label{ex:private}
We consider the protocol given
in~\cite{AbadiF04} designed for  transmitting a secret without revealing its
identity to other participants.
In this protocol, $A$ is willing to engage in
communication with~$B$ and  wants to reveal its identity to~$B$.
However, $A$ does not want to compromise its privacy by revealing its identity
or the identity of~$B$ more broadly. The participants~$A$ and~$B$ proceed as follows:\\[2mm]
\null\hfill
$ \begin{array}{rcl}
 A \rightarrow B& \; :\; & \{N_a,\pub_A\}_{\pub_B}\\
 B \rightarrow A& \;: \;& \{N_a,N_b,\pub_B\}_{\pub_A}
 \end{array}$
\hfill\null

\smallskip{}
Moreover, if the message received by $B$ is not of the expected form
then~$B$ sends out a ``decoy'' message: $\{N_b\}_{\pub_B}$. This message should basically look like~$B$'s other
 message from the point of view of an outsider.

\medskip{}

Relying on the signature and equational theory introduced in
Example~\ref{ex:signature}, a session of role~$A$ played by agent~$a$
(with private key $ska$)
with~$b$ (whose public key is $pkb$) can be modeled as follows:\\[1mm]
\null\hfill
$\begin{array}{lcl}
P(\ska, pkb) &\eqdef&
\Out(c_A, \aenc{\pair{n_a}{\pk{ska}}}{pkb}). \\
&& \In(c_A,x). \\ 
&& \testt{\pair{\projl{\adec{x}{ska}}}{\projr{\projr{\adec{x}{ska}}}} =
  \pair{n_a}{pkb}}{0}
\end{array}$\hfill\null
\smallskip{}
\\ \noindent
Here, we are only 
  considering the authentication protocol.
A more comprehensive model should include the access to an
application in case of a success.
Similarly, a session of role~$B$ played by agent~$b$ with~$a$ can be
modeled by the following basic proces where $N = \adec{y}{skb}$.\\[1mm]
\null\hfill
$\begin{array}{lcl}
Q(skb, pka) &\eqdef & \In(\cb,y) . \\
&& \texttt{if}\ \projr{N} = pka\;  \texttt{then} \; \Out(\cb, \aenc{\pair{\projl{N}}{\pair{n_b}{\pk{skb}}}}{pka}) \\
&&\phantom{\texttt{if}\ \projr{N} = pka}\; {\texttt{else} \; \Out(\cb, \aenc{n_b}{\pk{skb}})} 
\end{array}$\hfill\null

\smallskip{}

To model a scenario with one session of each role (played by the
agents~$a$ and~$b$), we may consider the extended process
$\proc{\p}{\Phi_0}$ where:
\begin{itemize}
\item  $\p
= \{P(ska, \pk{skb}), Q(skb, \pk{ska})\}$, and 
\item $\Phi_0 = \{w_0 \refer \pk{ska'}, w_1 \refer \pk{ska}, w_2
\refer \pk{skb}\}$. 
\end{itemize}
The purpose of $\pk{\ska'}$ will be clear later on. It allows us to 
consider
the existence of
another agent $a'$ whose public key  $\pk{ska'}$ is known by the attacker.
\end{example}

\subsection{Semantics}
\label{subsec:semantics}

We first define a standard concrete semantics. Thus, in this section, we work only with
closed extended processes, \ie processes $\proc{\p}{\Phi}$ where
$\fv(\p) = \varnothing$.

\bigskip{}

\noindent \resizebox{\textwidth}{!}{
$\begin{array}{lrcl}
 \mbox{\sc Then} \;\;\;\;\;&
  \multicolumn{3}{l}{\proc{\{\test{u=v}{Q_1}{Q_2}\}\uplus\p}{\Phi} 
\; \lrstep{\;\tau\;}\; \proc{\{Q_1\}\uplus\p}{\Phi} 
\hspace{1cm} \mbox{if $u =_{\E} v$}}\\[1.5mm]
\mbox{\sc Else} &
 \multicolumn{3}{l}{\proc{\{\test{u=v}{Q_1}{Q_2}\}\uplus\p}{\Phi}
 \; \lrstep{\;\tau\;} \; 
 \proc{\{Q_2\}\uplus\p}{\Phi} 
 \hspace{1cm} \mbox{if $u \neq_{\E} v$}}\\[1.5mm]
\mbox{\sc In} &
\proc{\{\In(c,x).Q\}\uplus\p}{\Phi} 
& \lrstep{\In(c,M)} & 
\proc{\{Q\{x\mapsto u\}\}\uplus\p}{\Phi}
\\ \multicolumn{4}{r}{\mbox{if $M \in \T(\dom(\Phi))$ and
    $M\Phi = u$}}
\\
 \mbox{\sc Out}
&
 \proc{\{\Out(c,u).Q\}\uplus\p}{\Phi}
 & \lrstep{\Out (c,w)} & \proc{\{Q\}\uplus\p}{\Phi\cup\{w\refer u\}}\\
\multicolumn{4}{r}{\mbox{if $w$ is a fresh variable}}
\\
\multicolumn{4}{l}{\text{where } c\in\mathcal{C}, w\in\mathcal{W}\text{ and }x\in\mathcal{X}.}\\
\end{array}
$}

\bigskip{}

A process may input any term that an attacker can build (rule {\sc
  In}): $\{x \mapsto u\}$ is a substitution that replaces any occurrence
of $x$ with $u$.
In the {\sc Out} rule, we enrich the attacker's knowledge by adding the newly 
output term~$u$, with a fresh handle~$w$, to the frame.
The two remaining rules are unobservable ($\tau$ action) from the point of 
view of the attacker.

\smallskip{}
%
 The relation $A \lrstep{a_1 \ldots a_k} B$
 between extended simple processes,
 where $k\geq 0$ and
 each $a_i$ is an observable or a $\tau$ action,
is defined in the usual way. We also consider the relation $\LRstep{\;\tr}$
defined as follows: $A \LRstep{\;\tr} B$ if, and only if, there exists
$a_1 \ldots a_k$ such that $A \lrstep{a_1 \ldots a_k} B$, and $\tr$ is
obtained from $a_1 \ldots a_k$ by erasing all occurrences of
$\tau$.


\begin{example}
\label{ex:semantics}
Consider the process $\proc{\p}{\Phi_0}$
introduced in Example~\ref{ex:private}. We have:\\[1mm]
\null\hfill
$\proc{\p}{\Phi_0}
 \lrstep{\Out(c_A,w_3) \cdot \In(\cb,w_3) \cdot \tau \cdot
   \Out(\cb,w_4) \cdot \In(c_A,w_4) \cdot \tau}
 \proc{\varnothing}{\Phi}$. \hfill\null

\smallskip{}

This trace corresponds to the normal execution of one instance of the
protocol. The two silent actions have been triggered using the {\sc
  Then} rule. The resulting frame $\Phi$ is as follows:\\[1mm]
\null\hfill
$\Phi_0 \uplus \{w_3 \refer
\aenc{\pair{n_a}{\pk{ska}}}{\pk{skb}}, \;w_4 \refer
\aenc{\pair{n_a}{\pair{n_b}{\pk{skb}}}}{\pk{ska}}\}.$\hfill\null
\end{example}


\subsection{Trace equivalence}
\label{subsec:equivalence}

Many interesting security properties, such as privacy-type
properties studied \eg in~\cite{arapinis-csf10}, are formalized
using the notion of \emph{trace equivalence}. 
We first introduce the notion of \emph{static
equivalence}  that compares sequences of messages. 

\begin{definition}[static equivalence]
Two frames $\Phi$ and $\Phi'$ are in \emph{static equivalence}, $\Phi
\statequiv \Phi'$, when we have that $\dom(\Phi) = \dom(\Phi')$,
and:\\[1mm]
\null\hfill
$M\Phi =_\E N\Phi \;\;\Leftrightarrow \;\; M\Phi' =_\E N\Phi'
\mbox{ for any terms $M, N \in \T(\dom(\Phi))$}.
$\hfill\null
\end{definition}

Intuitively, two frames are equivalent if an attacker cannot see the
difference between the two situations they represent, \ie   they satisfy the same equalities.

\begin{example}
\label{ex:static}
Consider the frame $\Phi$ given in Example~\ref{ex:semantics} and the
frame $\Phi'$ below:\\[1mm]
\null\hfill
$\Phi' \eqdef \Phi_0 \uplus \{w_3 \refer
\aenc{\pair{n_a}{\pk{ska'}}}{\pk{skb}}, \;\;w_4 \refer
\aenc{n_b}{\pk{skb}}\}.$
\hfill\null 

\smallskip{}

\noindent
Actually, we have that $\Phi \statequiv \Phi'$.  
Intuitively, the
equivalence holds since the attacker is not able to decrypt any of the
ciphertexts, and each ciphertext contains a nonce that prevents him to
build it from its components. 
Now, if we decide to give access to $n_a$ to the attacker, \ie considering 
$\Phi_+ = \Phi \uplus \{w_5 \refer n_a\}$ and $\Phi'_+ = \Phi' \uplus \{w_5
\refer n_a\}$, then the two frames $\Phi_+$ and $\Phi'_+$ are not in
static equivalence anymore. Let $M = \aenc{\pair{w_5}{w_1}}{w_2}$ and
$N = w_3$. We have that $M\Phi_+ =_{\E_\mathsf{aenc}} N\Phi_+$ whereas
$M\Phi'_+ \neq_{\E_{\mathsf{aenc}}} N\Phi'_+$.
\end{example}

\begin{definition}[trace equivalence]
Let~$A$ and~$B$ be two simple processes. We have that $A \sqsubseteq B$ if,
for every
sequence of actions $\tr$ such that $A \LRstep{\;\tr} \proc{\p}{\Phi}$, there
exists $\proc{\p'}{\Phi'}$ such that $B \LRstep{\;\tr}
\proc{\p'}{\Phi'}$ and $\Phi \statequiv \Phi'$. 
The processes~$A$ and~$B$ are 
\emph{trace equivalent}, denoted by $A \approx B$, if $A \sqsubseteq B$
and $B \sqsubseteq A$.
\end{definition}


\begin{example}
\label{ex:trace-equiv}
 Intuitively, the private authentication protocol presented in
 Example~\ref{ex:private} 
preserves anonymity if an attacker cannot
 distinguish whether $b$ is willing to talk to $a$ (represented by the
 process $Q(\skb,\pk{\ska})$) or willing to talk to $a'$  (represented by the
 process $Q(\skb,\pk{\ska'})$), provided $a$, $a'$ and~$b$ are honest participants.
This can be expressed relying on the following equivalence:\\[1mm]
\null\hfill
$\proc{ Q(\skb,\pk{\ska})}{\Phi_0} \stackrel{?}{\approx}
\proc{Q(\skb,\pk{\ska'})}{\Phi_0}.$\hfill\null

\smallskip{}

For illustration purposes,  we also consider a variant of the process~$Q$, denoted~$Q_0$, where its \texttt{else} branch has been
replaced by $\texttt{else } 0$. We will see that the ``decoy''
 message plays a crucial role to ensure privacy.
We have that:\\[1mm]
\null\hfill
$\proc{Q_0(\skb,\pk{\ska})}{\Phi_0} \lrstep{\In(\cb,\aenc{\pair{w_1}{w_1}}{w_2}) \cdot \tau
  \cdot  \Out(\cb,w_3)} \proc{\varnothing}{\Phi}$\hfill\null

\smallskip{}

\noindent where $\Phi = \Phi_0 \uplus \{w_3 \refer
\aenc{\pair{\pk{ska}}{\pair{n_b}{\pk{skb}}}}{\pk{ska}}\}$.

This trace has no counterpart in $\proc{Q_0(\skb,\pk{\ska'})}{\Phi_0}$. Indeed, we have
that:\\[1mm]
\null\hfill
$\proc{Q_0(\skb,\pk{\ska'})}{\Phi_0} \lrstep{\In(\cb,\aenc{\pair{w_1}{w_1}}{w_2}) \cdot \tau}
\proc{\varnothing}{\Phi_0}.$\hfill\null

\medskip{}

Hence, we
have that $\proc{Q_0(\skb,\pk{\ska})}{\Phi_0} \not\approx
\proc{Q_0(\skb,\pk{\ska'})}{\Phi_0}$.
Actually, it can been shown that $\proc{Q(\skb,\pk{\ska})}{\Phi_0} \approx
\proc{Q(\skb,\pk{\ska'})}{\Phi_0}$.
This is a non trivial equivalence that
can be checked using the tool APTE~\cite{APTE} within few seconds for a
simple scenario as the one considered here, and that takes few
minutes/days as soon as we want to consider 2/3 sessions of each
role.
\end{example}

%% file: compression.tex
\section{Reduction based on grouping actions}
\label{sec:compression}

A large number of possible interleavings results into multiple
occurrences of identical states. The compression step
lifts a common optimization that partly tackles this issue
in the case of reachability properties
to trace equivalence. The key idea is to force processes
to perform all enabled output actions as soon
as possible. In our setting, we can even safely force them
to perform a complete \emph{block} of input actions
followed by ouput actions.

\begin{example}
Consider the process $\proc{\p}{\Phi}$ with 
$\mathcal{P} = \{\In(c_1,x).P_1,\ \Out(c_2,b).P_2\}$.
In order to reach $\proc{\{P_1\{x \mapsto u\},\ P_2\}}{\Phi \cup \{w
  \refer b\}}$, we have to execute the action
$\In(c_1,x)$ (using a recipe $M$ that allows one to deduce $u$)
and the action $\Out(c_2,b)$ (giving us a label of the form $\Out(c_2,w)$). 
In case of reachability properties, the execution order of these
actions only 
matters if $M$ uses $w$.  Thus we can safely perform the outputs in priority.

The situation is more complex when considering trace equivalence.
In that case, we are concerned not only with reachable states, but also
with \emph{how} those states are reached. Quite simply, traces matter.
Thus, if we want to discard the trace $\In(c_1,M).\Out(c_2,w)$ when studying
process $\mathcal{P}$ and consider only its permutation
$\Out(c_2,w). \In(c_1,M)$, we have to 
make sure that the same permutation is available on the other process.
The key to ensure that identical permutations will be available on both
sides of the equivalence is our restriction to the class of simple processes.
\end{example}


\subsection{Compressed semantics}
\label{subsec:comp-semantics}

We now introduce the compressed semantics. Compression is an 
optimization, since it removes some interleavings. But it also gives rise to 
convenient ``macro-actions'', called \emph{blocks}, that combine a sequence of 
inputs followed by some outputs, potentially hiding silent actions.
Manipulating those blocks rather than indiviual actions
makes it easier to define our second optimization. 

\smallskip{}

For sake of simplicity, we consider \emph{initial} simple
processes. A simple
process $A = \proc{\p}{\Phi}$ is \emph{initial} if for any $P \in \p$,
we have that
$P = \In(c,x).P'$ for some channel~$c$, \ie each
basic process composing~$A$ starts with an input
action.

\begin{example}
\label{ex:initial}
Continuing Example~\ref{ex:private}, $\proc{\{P(\ska, \pk{\skb}),Q(\skb,\pk{\ska})\}}{\Phi_0}$ is not initial. 
Instead, we may consider $\proc{\{P_\init,
  Q(\skb,\pk{\ska})\}}{\Phi_0}$ where:\\[1mm]
\null\hfill
$
P_\init \eqdef \In(c_A,z). \testt{z = \start}{P(\ska,\pk{\skb})}$ \hfill\null
\smallskip{}

\noindent assuming that $\start$ is a (public) constant in our signature.
\end{example}

\begin{figure}[h]
  \[
  \begin{array}{lc}
\mbox{\sc In} &  \infer[\mbox{with }\ell\in\{i^*;i^+\}]
  {\proc{P}{\Phi}\sintf{\;\In(c,M). \tr\;}{\ell} \proc{P''}{\Phi''}}
  {\proc{P}{\Phi}\lrstep{\In(c,M)} \proc{P'}{\Phi'}
  & \proc{P'}{\Phi'}\sintf{\;\;\tr\;\;}{i^*} \proc{P''}{\Phi''}}
   \\ [2mm]
\mbox{\sc Out} & 
 \infer[\mbox{with } \ell\in\{i^*;o^*\}]
   {\proc{P}{\Phi}\sintf{\;\Out(c,w). \tr\;}{\ell} \proc{P''}{\Phi''}}
   {\proc{P}{\Phi}\sint{\Out(c,w)} \proc{P'}{\Phi'}
   & \proc{P'}{\Phi'}\sintf{\;\;\tr\;\;}{o^*} \proc{P''}{\Phi''}}
 \\ [2mm]
\mbox{\sc Tau} &
 \infer[\mbox{with }\ell\in\{o^*;i^+;i^*\}]
   {\proc{P}{\Phi} \sintf{\;\;\tr\;\;}{\ell} \proc{P''}{\Phi''}}
   {\proc{P}{\Phi} \lrstep{\;\;\tau\;\;} \proc{P'}{\Phi'} &
    \proc{P'}{\Phi'} \sintf{\;\;\tr\;\;}{\ell} \proc{P''}{\Phi''}}
 \\ [2mm]
\mbox{\sc Proper} &
   \infer{{\proc{0}{\Phi}\sintf{\;\;\epsilon\;\;}{o^*}
   \proc{0}{\Phi}}}{}
   \quad
   \infer{
     {\proc{\In(c,x).P}{\Phi}\sintf{\;\;\epsilon\;\;}{o^*}
     \proc{\In(c,x).P}{\Phi}}
   }{}
 \\ [2mm]
 \mbox{\sc Improper} &
 \infer{
   {\proc{0}{\Phi}\sintf{\;\;\epsilon\;\;}{i^*} \proc{\bot}{\Phi}}
 }{}
\end{array}
\]
  \caption{Focused semantics on extended basic processes}
  \label{fig:sintf}
\end{figure}

The main idea of the compressed semantics is to ensure
that when a basic process starts executing some actions, it actually
executes a maximal block of actions. In analogy with focusing in sequent
calculus, we say that the basic process takes the focus, and can only
release it under particular conditions.
We define in Figure~\ref{fig:sintf} how blocks can be executed by
extended basic processes. In that semantics,
the label~$\ell$ denotes the stage of the execution, starting with $i^+$, then
$i^*$ after the first input and $o^*$ after the first output.

\begin{example}
\label{ex:compressed-semantics}
Going back to Example~\ref{ex:trace-equiv}, we have that:\\[1mm]
\null\hfill
$
\proc{Q_0(\skb,\pk{\ska})}{\Phi_0}
\sintf{\;\;
  \In(\cb,\aenc{\pair{w_1}{w_1}}{w_2})\cdot\Out(\cb,w_3) \;\;}{i^+}
\proc{0}{\Phi}$\hfill\null

\smallskip{}

\noindent  where $\Phi$ is as given in Example~\ref{ex:trace-equiv}.
As illustrated by the prooftree below, we have also $\proc{Q_0(\skb,\pk{\ska})}{\Phi_0}
\sintf{\;\tr\;}{i^+} \proc{\bot}{\Phi_0}$ with $\tr =
\In(\cb,\aenc{\pair{w_1}{w_1}}{w_2})$. 
\smallskip{}
\noindent \resizebox{\textwidth}{!}{
$
\begin{prooftree}
\proc{Q_0(\skb,\pk{\ska})}{\Phi_0} \lrstep{\tr}
\proc{Q'}{\Phi_0}
\begin{prooftree}
\proc{Q'}{\Phi_0} \lrstep{\tau} \proc{0}{\Phi_0} 
\;\;
\begin{prooftree}
\justifies
\proc{0}{\Phi_0} \sintf{\epsilon}{i^*} \proc{\bot}{\Phi_0}
\using{\mbox{\sc {Improper}} \hspace{-0.7cm}}
\end{prooftree}
\justifies
\proc{Q'}{\Phi_0} \sintf{\epsilon}{i^*} \proc{\bot}{\Phi_0}
\using{\mbox{\sc {Tau}} \hspace{-0.4cm}} 
\end{prooftree}
\justifies
\proc{Q_0(\skb,\pk{\ska})}{\Phi_0}
\sintf{\tr}{i^+} \proc{\bot}{\Phi_0}
\using{\mbox{\sc In}}
\end{prooftree}
$}
\noindent where $Q' \eqdef \testt{\pk{\ska} = \pk{\ska}}{\Out(c_B,u)}$ for some
message $u$. 
\end{example}

Then we define the compressed reduction $\sintc{}$
between extended simple processes
{as the least reflexive transitive relation satisfying the
following rules}:

\bigskip{}

\noindent \resizebox{\textwidth}{!}{
$\begin{array}{lclclcl}
 \mbox{\sc Block} &&
  \infer
  {\proc{\{Q\}\uplus \p}{\Phi}\sintc{\;\;\tr\;\;}
    \proc{\{Q'\}\uplus\p}{\Phi'}}
  {\proc{Q}{\Phi}\sintf{\;\;\tr\;\;}{i^+} \proc{Q'}{\Phi'}
    & Q'\neq\bot} 
&\;&
\mbox{\sc Failure} &&
  \infer
  {\proc{\{Q\}\uplus\p}{\Phi}\sintc{\;\;\tr\;\;} \proc{\varnothing}{\Phi'}}
  {\proc{Q}{\Phi}\sintf{\;\;\tr\;\;}{i^+} \proc{Q'}{\Phi'} & Q' = \bot}
  \end{array}$
}

\bigskip{}

A basic process is allowed to \emph{properly} end a block execution when it 
has performed outputs and it cannot perform any more.
Accordingly, we call \emph{proper block} a non-empty 
sequence of inputs followed by a non-empty sequence of outputs, all on the 
same channel.
For completeness, we also allow improper termination of a block,
when the basic process
that is currently executing is not able to perform any visible action
(input or output) and it has not yet performed an output.

\begin{example}
Continuing Example~\ref{ex:compressed-semantics}, 
using the rule {\sc block}, we can derive that:\\[1mm]
\null\hfill
$\proc{\{P_\init, Q_0(\skb,\pk{\ska})\}}{\Phi_0}
\sintc{\;\;
  \In(\cb,\aenc{\pair{w_1}{w_1}}{w_2})\cdot\Out(\cb,w_3) \;\;}
\proc{P_\init}{\Phi}.$\hfill\null

\smallskip{}

\noindent We can also derive 
$\proc{\{P_\init, Q_0(\skb,\pk{\ska'})\}}{\Phi_0}
\sintc{\;\;
  \In(\cb,\aenc{\pair{w_1}{w_1}}{w_2})\;\;}
\proc{\varnothing}{\Phi_0}$ (using the rule {\sc Improper}).
Note that the resulting simple process is reduced to~$\varnothing$ even though
$P_\init$ has never been executed.
\end{example}

At first sight, killing the whole process when applying the rule {\sc Improper} may seem too strong. 
Actually, even if this kind of scenario is observable by the attacker, it does not bring
him any new knowledge, hence it plays only a limited role:
it is in fact sufficient to consider such improper blocks
at the end of traces.

 \begin{example}
 \label{ex:dirty-block}
 Consider $\p = \{ \In(c,x).\In(c,y),\; \In(c',x') \}$.
 Its compressed traces are of the form $\In(c,M).\In(c,N)$ and $\In(c',M')$.
 The concatenation of those two improper traces cannot be executed in 
 the compressed semantics. Intuitively, we do not loose anything for
 trace equivalence, because if a process can exhibit those two improper
 blocks they must be in parallel and hence considering their combination is 
 redundant.
 \end{example}


We define the notion of \emph{compressed trace equivalence} (resp. \emph{inclusion}) accordingly
relying on $\sintc{}$ instead of $\LRstep{}$,
and we denote them $\approx_c$ (resp. $\sqsubseteq_c$). 


\subsection{Soundness and completeness}
\label{subsec:soundness-completeness}

 The purpose of this section is to establish the soundness and
 completeness of the compressed semantics. More precisely, we show that
 the two relations~$\approx$ and~$\approx_c$ coincide on initial simple
 processes.
%

Intuitively, we can always permute output  (resp. input) actions occurring
on distinct channels, and we can also permute an output
with an input if the outputted message is not used to build the inputted
term.
More formally, 
we define an \emph{independence relation~$\Ind_a$ over actions} as
the least symmetric relation satisfying:
\vspace{-2pt}                   
\begin{itemize}
\item $\Out(c_i,w_i) \Ind_a \Out(c_j,w_j)$ and $\In(c_i,M_i) \Ind_a
  \In(c_j,M_j)$ as soon as $c_i \neq c_j$, 
\item $\Out(c_i,w_i) \Ind_a \In(c_j,M_j)$ when in addition $w_i
  \not\in \fv(M_j)$.
\end{itemize}
Then, we consider ${=}_{\Ind_a}$ to be the least
congruence (w.r.t. concatenation) satisfying $\act\cdot\act' =_{\Ind_a} \act'\cdot \act$
for all $\act$ and $\act'$ with $\act \Ind_a \act'$, and we show that
processes are equally able to execute equivalent (w.r.t. $=_{\Ind_a}$) traces.

\begin{restatable}{lemma}{lempermute}
\label{lem:permute}
Let $A$, $A'$ be two simple extended processes and $\tr$, $\tr'$ be such that
$\tr =_{\Ind_a} \tr'$. We have that $A \LRstep{\;\tr} A'$ if, and only if, $A
\LRstep{\;\tr'} A'$. 
\end{restatable}



Now, considering traces that are only made of proper blocks, a
strong relationship can be established between the two semantics.

\begin{restatable}{proposition}{procompsoundness}
\label{pro:comp-soundness}
Let $A$, $A'$ be two simple extended processes, and $\tr$ be a trace
made of proper blocks  such that $A \sintc{\;\tr\;} A'$. Then we have
that  $A \LRstep{\;\tr} A'$.
\end{restatable}

\begin{restatable}{proposition}{procompcompleteness}
\label{pro:comp-completeness}
Let $A$, $A'$ be two initial simple processes, and $\tr$ be a trace
made of proper blocks such that
 $A \LRstep{\;\tr} A'$.
Then, we have that $A \sintc{\;\tr\;} A'$.
\end{restatable}

\begin{restatable}{theorem}{theocompsoundnesscompleteness}
\label{theo:comp-soundness-completeness}
Let $A$ and $B$ be two initial simple processes. We have that:\\[1mm]
\null\hfill
$A \approx B \;\; \Longleftrightarrow\;\; A \approxc B.$\hfill\null
\end{restatable}

\begin{proof}\emph{(Sketch)}
The main difficulty 
is that Proposition~\ref{pro:comp-completeness} only considers traces
composed of proper blocks whereas we have to consider all traces.
To prove the $\Rightarrow$ implication, we
have to pay attention to the last block of the compressed trace that
can be an improper one (composed of several inputs on a channel $c$). 
The $\Leftarrow$ implication is more difficult since we have to
consider a trace $\tr$ of a process~$A$
that is an interleaving of some prefix of proper and
improper blocks. We will first complete it to obtain an interleaving of
complete blocks and improper blocks. We then reorganize the 
actions providing an equivalent trace $\tr'$  w.r.t. $=_{\Ind_a}$
such that $\tr' = \tr_{\mathrm{io}} \cdot \tr_{\mathrm{in}}$ where $\tr_{\mathrm{io}}$
is made of proper blocks and $\tr_{\mathrm{in}}$ is made of improper blocks.
For each improper block~$b$ of $\tr_{\mathrm{in}}$, we show by applying
Lemma~\ref{lem:permute} and Proposition~\ref{pro:comp-completeness}
that $A$ is able to perform $\tr_{\mathrm{io}} \cdot b$ in the compressed
semantics and thus $B$ as well.
Finally, we show that the executions of all those (concurrent) blocks~$b$
can be put together, obtaining that~$B$ can perform~$\tr'$.
\qed
\end{proof}

Note that, as illustrated by the following example, the two underlying
notions of trace 
inclusion do \emph{not} coincide.

\begin{example}
Let $P = \In(c,x)$ and $Q = \In(c,x).\Out(c,n)$. 
Actually, we have that 
${\proc{P}{\varnothing} \sqsubseteq \proc{Q}{\varnothing}}$ whereas
%
$\proc{P}{\varnothing} \not\sqsubseteq_c \proc{Q}{\varnothing}$
since in the compressed semantics~$Q$ is not
allowed to stop its execution after its first input. 
\end{example}

%% file: symbolic.tex

\section{Deciding trace equivalence via constraint solving}
\label{sec:constraint-solving}

In this section, we propose a symbolic semantics for our compressed
semantics following, \eg \cite{MS02,baudet-ccs2005}. 
Such a semantics avoids
potentially infinite branching of our reduction semantics
due to inputs from the environment. 
Correctness is maintained by associating with each process a set of constraints on terms.

\subsection{Constraint systems}
\label{subsec:constraint-systems}

Following the notations of~\cite{baudet-ccs2005}, we consider a new set $\X^2$ of 
\emph{second-order variables}, denoted by $X$, $Y$, etc.
 We shall use those variables to abstract over recipes.
We denote by $\fvs(o)$ the set of free second-order variables of an object
$o$, typically a constraint system. To prevent ambiguities, we shall
use $\fvp$ instead of $\fv$ for free first-order variables.

\begin{definition}[constraint system]\label{def:cs}
 A \emph{constraint system} $\C = \cs{\Phi}{\Set}$ consists of a frame
 $\Phi$, and a set of
  constraints $\Set$.
 We consider three kinds of \emph{constraints}:\\[1mm]
\null\hfill $\dedi{D}{X}{x} \quad\quad u \eqi v \quad\quad u \neqi v$
\hfill\null
\smallskip{}

\noindent  where $D\subseteq\W$, $X\in\X^2$, $x\in\X$ and $u,v\in\T(\N\cup\X)$.
\end{definition}

The first kind of constraint expresses that a recipe~$X$
has to use only 
variables from a certain set~$D$, and
that the obtained term should be $x$. The handles in~$D$
 represent terms that have been previously outputted by the process.

We are not interested in general constraint systems, but only
consider constraint systems that are \emph{well-formed}.
Given $\C$, we define a dependency order on $\fvp(\C) \cap \X$
  by declaring that $x$ depends on $y$ if, and only if, $\Set$ contains a
  deduction constraint $\dedi{D}{X}{x}$ with $y\in\fvp(\Phi(D))$.
  For $\C$ to be a well-formed constraint system, we require that the
  dependency relationship is acyclic and that for every
  $x\in\fvp(\C) \cap \X$ (resp.~$X\in\fvs(\C)$) there is a unique constraint
  $\dedi{D}{X}{x}$ in $\Set$.
  For $X\in\fvs(\C)$, we write $D_\C(X)$ for the domain $D\subseteq\W$
  of the deduction constraint $\dedi{D}{X}{x}$ associated to~$X$
  in~$\C$.

\begin{example}
\label{ex:cs}
Let $\Phi = \Phi_0 \uplus \{w_3 \refer
\aenc{\pair{\projr{N}}{\pair{n_b}{\pk{skb}}}}{\pk{ska}}\}$ with $N =
\adec{y}{skb}$, and $\Set$ be a set containing two constraints:\\[1mm]
\null\hfill
$\dedi{\{w_0,w_1,w_2\}}{Y}{y} \mbox{ and } \projr{N} \eqi
\pk{ska}.$\hfill\null

\smallskip{}

\noindent We have that $\C = \cs{\Phi}{\Set}$ is a well-formed constraint system.
There is only one first-order variable $y \in \fvp(\C) \cap \X$, and it does
not occur in $\fvp(\Phi(\{w_0,w_1,w_2\}))$, which is empty. Moreover,
there is indeed a unique constraint that introduces~$y$.
\end{example}

Our notion of well-formed constraint systems is in line with 
what is used
\eg in~\cite{MS02,baudet-ccs2005}. 
We use a simpler and
(slightly) more permissive variant because we are not concerned with
constraint solving procedures in this work.

\begin{definition}[solution] \label{def:cssol}
A \emph{solution} of a constraint system $\C = \cs{\Phi}{\Set}$ is a
substitution $\theta$ such that
$\dom(\theta) = \fvs(\C)$, and
$X\theta \in \T(D_\C(X))$ for any $X \in \dom(\theta)$.
Moreover, we require that there exists a ground substitution $\lambda$
with $\dom(\lambda) = \fvp(\C)$ such that:
\begin{itemize}
  \item for every $\dedi{D}{X}{x}$ in $\Set$, we have that
  $(X\theta)(\Phi\lambda) =_\E x\lambda$;
\item for every $u \eqi v$ in $\Set$, we have that $u\lambda =_\E
  v\lambda$; and
\item for every $u \neqi v$ in $\Set$, we have that $u\lambda \neq_\E
  v\lambda$.
\end{itemize}
The set of solutions of a constraint system $\C$ is denoted $\Sol(\C$).
Since we consider constraint systems that are well-formed,
the substitution $\lambda$ is unique modulo $\E$ given 
$\theta\in\Sol(\C)$. We denote it by $\lambda_\theta$ when $\C$ is
clear from the context.
\end{definition}

\begin{example}
\label{ex:solution}
Consider again the constraint system $\C$ given in
Example~\ref{ex:cs}. We have that $\theta = \{Y \mapsto
\aenc{\pair{w_1}{w_1}}{w_2}\}$ is a solution of $\C$. Its
associated first-order solution is $\lambda_\theta = \{y \mapsto
\aenc{\pair{\pk{ska}}{\pk{ska}}}{\pk{skb}}\}$. 
\end{example}


\subsection{Symbolic processes: syntax and semantics}
\label{subsec:symbolic-calculus}

From a simple process $\proc{\p}{\Phi}$, 
we compute the constraint systems capturing its possible executions,
starting from the symbolic process $\triple{\p}{\Phi}{\varnothing}$.
Note that we are now manipulating processes that rely on free
variables.

\begin{definition}[symbolic process] \label{def:symbolic-process}
A symbolic process is a tuple $\triple{\p}{\Phi}{\Set}$ where
$\cs{\Phi}{\Set}$ is a constraint system and
$\fvp(\p)\subseteq(\fvp(\Set) \cap \X)$.
\end{definition}

\noindent We give below 
a standard symbolic semantics for 
our symbolic processes.

\medskip{}

\noindent\resizebox{\textwidth}{!}{$
\begin{array}{lcl}
\mbox{\sc In} &\;\;\;\;&
\triple{\In(c,y).P}{\Phi}{\Set} \ssym{\In(c,X)}
\triple{P\{y\mapsto x\}}{\Phi}{\Set \cup \{ \dedi{\dom(\Phi)}{X}{x} \}}\\
&& \hspace{1cm}\hfill\mbox{where $X$ (resp. $x$) is a  fresh second-order
  (resp. first-order) variable} 
\\
\mbox{\sc Out} &&
\triple{\Out(c,u).P}{\Phi}{\Set} \ssym{\Out(c,w)}
  \triple{P}{\Phi\cup\{w\refer u\}}{\Set}\\
&&\hfill\mbox{where $w$ is a fresh first-order variable}
\\
\mbox{\sc Then} &&
\triple{\test{u=v}{P}{Q}}{\Phi}{\Set} \ssym{\;\;\tau\;\;}
\triple{P}{\Phi}{\Set \cup \{ u \eqi v \}}
\\[1mm]
\mbox{\sc Else} &&
\triple{\test{u=v}{P}{Q}}{\Phi}{\Set} \ssym{\;\;\tau\;\;}
\triple{Q}{\Phi}{\Set \cup \{ u \neqi v \}}
\end{array}
$}

\medskip{}


From this semantics, we derive our compressed symbolic 
semantics $\ssymc{\;\tr\;}$ following 
the same pattern as for the concrete semantics (see
Section~\ref{subsec:app-symbolic-semantics} in Appendix).
We consider interleavings that execute maximal blocks of
actions, and we allow improper termination of a block only at
the end of a trace.
\begin{example}
\label{ex:symbolic-semantics}
We have that $\triple{Q_0(b,a)}{\Phi_0}{\varnothing} \ssymc{\;\tr\;}
\triple{\varnothing}{\Phi}{\Set}$ where:
\vspace{-1mm}
\begin{itemize}
\item $\tr = \In(\cb,
Y
)\cdot\Out(\cb,w_3)$, and
\item $\C = \cs{\Phi}{\Set}$ is the constraint system defined in
  Example~\ref{ex:cs}.
\end{itemize}
\end{example}

We are now able to define our notion
of (symbolic) trace equivalence.

\begin{definition}[trace equivalence w.r.t. $\ssymc{\tr}$]
\label{def:equiv-comp-symb}
Let~$A = \proc{\p}{\Phi}$
and~$B=\proc{\q}{\Psi}$ be two simple processes. We have that $A \sqsubseteq_s B$
when, for every sequence $\tr$ such
  that $\triple{\p}{\Phi}{\varnothing} \ssymc{\;\tr}
  \triple{\p'}{\Phi'}{\Set_A}$, for every $\theta \in
  \Sol(\Phi';\Set_A)$, we have that:
\vspace{-1mm}
\begin{itemize}
\item $\triple{\q}{\Psi}{\varnothing} \ssymc{\;\tr}
  \triple{\q'}{\Psi'}{\Set_B}$ with $\theta \in
  \Sol(\Psi';\Set_B)$, and 
\item $\Phi\lambda^A_\theta \statequiv
  \Psi\lambda^B_\theta$ where $\lambda^A_\theta$ (resp. $\lambda^B_\theta$) is the substitution
  associated to $\theta$ w.r.t. $\cs{\Phi'}{\Set_A}$
  (resp. $\cs{\Psi'}{\Set_B}$).
\end{itemize}
\vspace{-1mm}
We have that $A$ and $B$ are in \emph{trace equivalence
  w.r.t. $\ssymc{\tr}$} if $A \sqsubseteq_s B$ and $B
\sqsubseteq_s A$.
\end{definition}

\begin{example}
We have that $\proc{Q_0(b,a)}{\Phi_0} \not\sqsubseteq_s
\proc{Q_0(b,a')}{\Phi_0}$. Continuing
Example~\ref{ex:symbolic-semantics}, we have seen that $\triple{Q_0(b,a)}{\Phi_0}{\varnothing} \ssymc{\;\tr\;}
\triple{\varnothing}{\Phi}{\Set}$, and $\theta \in \Sol(\Phi;\Set)$
(see Example~\ref{ex:cs}). The only symbolic process that
is reachable from $\triple{Q_0(b,a')}{\Phi_0}{\varnothing}$ 
using $\tr$ is $\triple{\varnothing}{\Phi'}{\Set'}$ with:
\vspace{-1mm}
\begin{itemize}
\item $\Phi' = \Phi_0 \uplus\{w_3 \refer \aenc{\pair{\projr{N}}{\pair{n_b}{\pk{skb}}}}{\pk{ska'}}\}$, and
\item $\Set' = \big\{\dedi{\{w_0,w_1,w_2\}}{Y}{y};\;\; \projr{N} \eqi
  \pk{ska'}\big\}$.
\end{itemize} 
\vspace{-1mm}
We can check that $\theta$ is not a solution of $\cs{\Phi'}{\Set'}$.
\end{example}

For processes without replication, the symbolic transition system is
finite. Thus, deciding (symbolic) trace equivalence between
processes boils down to
the problem of deciding a notion of equivalence between sets of
constraint systems. This problem is well-studied and several
procedures already exist~\cite{baudet-ccs2005,chevalier10,cheval-ccs2011}.


\subsection{Soundness and completeness}
\label{subsec:symbolic-soundness-completeness}

Using the usual approach, such as the one developed in~\cite{baudet-ccs2005,CCD-tcs13}, we can
show soundness and completeness of our symbolic compressed semantics
w.r.t. our concrete compressed semantics. We have:
\begin{itemize}
\item \emph{Soundness}: each transition in the compressed symbolic semantics represents a set
of transitions that can be done in the concrete compressed semantics.
\item \emph{Completeness}: each transition in the compressed semantics can be
matched by a transition in the compressed symbolic semantics.
\end{itemize}

Finally, relying on these two results, we can establish  
that symbolic trace equivalence ($\approx_s$) exactly captures
compressed trace equivalence ($\approx_c$). 

\begin{restatable}{theorem}{thrmsintcssymv}
 \label{thrm:sintc-ssymc}
For any extended simple processes $A$ and $B$, we have that:\\[1mm]
\null\hfill
 $A \inceintc B
\iff A \incesym B$.
\hfill\null
\end{restatable}

%% file: differentiation.tex
\section{Reduction using dependency constraints}
\label{sec:diff}

Unlike compression, 
which is based only on the input/output nature of actions, our second
optimization takes
into account the exchanged messages.

Let us first illustrate one simple instance
of our optimization and how dependency constraints~\cite{ModersheimVB10}
may be used to incorporate it in symbolic semantics.
Let $P_i = \In(c_i,x_i).\Out(c_i,u_i).P'_i$ with $i \in \{1,2\}$,
and $\Phi_0 = \{w_0 \refer n\}$ be a closed frame.
We consider the simple process $A = \proc{\{P_1,P_2\}}{\Phi_0}$, and
the  two symbolic interleavings depicted below.

\begin{wrapfigure}[9]l{6.0cm}
\small
\vspace{-0.3cm}
\begin{tikzpicture}[inner sep=0pt,node distance=0.7cm]


\node (sommet) {$\bullet$};
\node[below of=sommet, left=1.2cm of sommet] (l1) {$\bullet$};
\node[below of=l1] (l2) {$\bullet$};
\node[below of=l2] (l3) {$\bullet$};
\node[below of=l3] (l4) {$\bullet$};
\node[below of=sommet, right=1.2cm of sommet] (r1) {$\bullet$};
\node[below of=r1] (r2) {$\bullet$};
\node[below of=r2] (r3) {$\bullet$};
\node[below of=r3] (r4) {$\bullet$};


\path (sommet) edge  node[left,pos=0.3]{$\In(c_1,X_1)\;\;\;\;$} (l1);
\path (l1) edge node[left]{$\Out(c_1,w_1)\;$} (l2);
\path (l2) edge node[left]{$\In(c_2,X_2)\;$} (l3);
\path (l3) edge node[left]{$\Out(c_2,w_2)\;$} (l4);
\path (sommet) edge  node[right,pos=0.3]{$\;\;\;\;\In(c_2,X_2)$} (r1);
\path (r1) edge node[right]{$\;\Out(c_2,w_2)$} (r2);
\path (r2) edge node[right]{$\;\In(c_1,X_1)$} (r3);
\path (r3) edge node[right]{$\;\Out(c_1,w_1)$} (r4);


\coordinate (l4t) at ($ (l4) - (90:0.1cm) $);
\coordinate (l4l) at ($ (l4t) - (60:0.5cm) $);
\coordinate (l4r) at ($ (l4t) - (120:0.5cm) $);
\filldraw[color=gray!20] (l4t) -- (l4l) -- (l4r) -- (l4t);
\coordinate (r4t) at ($ (r4) - (90:0.1cm) $);
\coordinate (r4l) at ($ (r4t) - (60:0.5cm) $);
\coordinate (r4r) at ($ (r4t) - (120:0.5cm) $);
\filldraw[color=gray!20] (r4t) -- (r4l) -- (r4r) -- (r4t);


\node[right=1.0cm of l4] (cl) {} ;
\node[left=1.0cm of r4] (cr) {} ;
\filldraw[pattern=north east lines] (cl) circle (0.5);
\filldraw[pattern=north west lines] (cr) circle (0.5);
\draw[dashed] (l4) -- ($ (cl) + (120:0.5cm) $);
\draw[dashed] (l4) -- ($ (cl) + (-120:0.5cm) $);
\draw[dashed] (r4) -- ($ (cr) - (120:0.5cm) $);
\draw[dashed] (r4) -- ($ (cr) - (-120:0.5cm) $);

\end{tikzpicture}
\end{wrapfigure}
\noindent
The two resulting symbolic processes are of the form
$\triple{\{P'_1,P'_2\}}{\Phi}{\Set_i}$ where: 
\begin{itemize}
\item 
$\Phi = \Phi_0 \uplus \{w_1 \refer u_1, w_2 \refer u_2\}$, 
\item
$\Set_1 = \big\{\dedi{w_0}{X_1}{x_1}; \; \dedi{w_0,
  w_1}{X_2}{x_2}\big\}$, 
\item 
$\Set_2 = \big\{\dedi{w_0}{X_2}{x_2};\; \dedi{w_0,
  w_2}{X_1}{x_1}\big\}$. 
\end{itemize}

The sets of concrete processes that these two symbolic processes
represent are different, which means that we cannot discard any of
those interleavings. However, these sets have a significant overlap
corresponding to concrete instances of the interleaved blocks that
are actually independent, \ie where the output of one block is not
necessary to obtain the input of the next block.
In order to avoid considering such concrete processes twice,
we may add a \emph{dependency constraint} $\constrd{X_1}{w_2}$ in $\C_2$,
whose purpose is to discard all solutions $\theta$
such that the message $x_1\lambda_\theta$ can be derived without using
$w_2 \refer u_2\lambda_\theta$.
For instance, the concrete trace 
$\In(c_2, w_0)\cdot \Out(c_2,w_2)\cdot \In(c_1,w_0)\cdot \Out(c_1,w_1)$
would be discarded thanks to this new constraint.

The idea of \cite{ModersheimVB10} is to accumulate dependency constraints
generated whenever such a pattern is detected in an execution, and use
an adapted constraint resolution procedure to narrow and eventually
discard the constrained symbolic states.
We seek to exploit similar ideas for optimizing the verification of
trace equivalence rather than reachability.
This requires extra care, since pruning traces as
described above may break completeness when considering trace equivalence. As 
before, the key to obtain a valid optimization will be to discard traces in a 
similar way on the two processes being compared.
In addition to handling this necessary subtlety, we also propose a new proof
technique for justifying dependency constraints. The generality of that 
technique allows us to add more dependency constraints, taking into account
more patterns than the simple diamond shape from the previous example.

There are at least two natural semantics for dependency constraints.
The simplest semantics focuses on the second-order notion of recipe. In the
above example, it would require that recipe $X_1\theta$ contains the variable
$w_2$. That is weaker than a first-order semantics requiring that \emph{any} 
recipe derivinig $x_1\lambda_\theta$ would involve $w_2$ since spurious
dependencies may easily be introduced.
%
Our ultimate goal in this section is to show that trace equivalence w.r.t.
the first-order reduced semantics coincides with the regular symbolic 
semantics. However, we first establish this result for the second-order
semantics, which is more easily analysed and provides a useful stepping
stone.

\subsection{Second-order reduced semantics}
\label{subsec:reduced-semantics}

We start by introducing \emph{dependency constraints}, in a more general form
than the one used above, and give them a second-order semantics.

\begin{definition}[dependency constraint]
  A \emph{dependency
  constraint} is a constraint of the form $\constrd{\vect X}{\vect w}$
  where $\vect X$ is a vector of variables in $\X^2$, 
  and $\vect w$ is a vector of handles, i.e. variables in $\W$.

Given a substitution $\theta$ with $\dom(\theta) \subseteq \X^2$, and
$X\theta \in \T(\W)$ for any $X \in \dom(\theta)$. 
We say that $\theta$ satisfies
$\constrd{\vect X}{\vect w}$, denoted $\theta \models  \constrd{\vect
  X}{\vect w}$, if  either $\vect w = \varnothing$ or there exist
  $X_i\in\vect{X}$ and $w_j\in\vect w$ such that
  $w_j\in\fvp(X_i\theta)$.
\end{definition}

A constraint system with dependency constraints is called a \emph{dependency
constraint system}.
We denote by $\noconstrd{\C}$ the regular constraint system 
obtained by removing all dependency constraints from $\C$.
We only consider \emph{well-formed} dependency constraint systems, that is
those $\C$ such that $\noconstrd{\C}$ is well-formed.
A solution of~$\C$
is a substitution~$\theta$ such that $\theta \in \Sol(\noconstrd{\C})$ and $\theta
\models  \constrd{\vect X}{\vect w}$ for each dependency constraint $\constrd{\vect X}{\vect w} \in \C$.
We denote this set $\Sol^2(\C)$. 





We shall now define how dependency constraints will be added to
our constraint systems. For this, we fix an arbitrary
total order $\och$ on channels. Intuitively, this order expresses
which executions should be favored, and which should be allowed
only under dependency constraints.
To simplify the presentation, we use the notation $\io{c}{X}{w}$ as a
shortcut for $\In(c,X_1)\cdot \ldots \cdot \In(c,X_\ell). \Out(c,w_1)
\cdot \ldots \cdot \Out(c,w_k)$ assuming that $\vect{X} =
(X_1,\ldots,X_\ell)$ and $\vect{w} = (w_1,\ldots,w_k)$. Note that
$\vect{X}$ and/or $\vect{w}$ may be empty.

\begin{definition}[generation of dependency constraints]
Let $c$ be a channel, and
 $\tr = \io{c_1}{X_1}{w_1}\cdot\ldots\cdot\io{c_{n}}{X_{n}}{w_{n}}$
 be a trace.
 If there exists a rank $k\leq n$ such that
 $c_i \och c \och c_k$ for all $k < i \leq n$, then we define
 \[
 \Dep{\tr}{c} =\{\; w ~|~ w \in \vect{w_i} \mbox{ with $k \leq i
     \leq n$}\}
 \]
 Otherwise, we have that $\Dep{\tr}{c} = \varnothing$.
\end{definition}

We obtain our reduced semantics by integrating those dependency constraints
into the symbolic compressed semantics. We define $\ssymd{}$ as the least
reflexive relation satisying the following rule:
\[
\begin{array}{c}
\infer[]
{{\triple{\p}{\Phi}{\varnothing} \ssymd{\tr\cdot\io{c}{X}{w}}
    \triple{\p''}{\Phi''}{\Set'' \cup \{\constrd{\vect X}{\Dep{\tr}{c}}\}}}}
{\triple{\p}{\Phi}{\varnothing} \ssymd{\;\tr\;}
  \triple{\p'}{\Phi'}{\Set'} & \triple{\p'}{\Phi'}{\Set'}
  \ssymc{\io{c}{X}{w}} \triple{\p''}{\Phi''}{\Set''}}
\end{array}
\]


\medskip{}


Given a proper trace, we define $\AllDep{\tr}$ to be the accumulation of the
generated constraints as defined above for all prefixes of $\tr$
(where each proper block is considered as an atomic action).
We may observe that:
\begin{itemize}
\item  if $A \ssymd{\;\tr\;} \triple{\p}{\Phi}{\Set}$ then
$\Set = \noconstrd{\Set} \cup \AllDep{\tr}$ and
$A \ssymc{\;\tr\;}\triple{\p}{\Phi}{\noconstrd{\Set}}$;
\item  if $A \ssymc{\;\tr\;} \triple{\p}{\Phi}{\Set}$ then 
$\Set = \noconstrd{\Set'}$ and $A \ssymd{\;\tr\;} \triple{\p}{\Phi}{\Set'}$.
\end{itemize}

\begin{example}
\label{ex:taille-trois}
Let $a$, $b$, and $c$ be channels in $\C$ such that $a \och b \och c$.
The dependency constraints generated during the symbolic execution of a simple process of the form
$\proc{\{\In(a,x_a).\Out(a,u_a), \,\In(b,x_b).\Out(b,u_b),  \,\In(c,x_c).\Out(c,u_c)\}}{\Phi}$
are depicted below.
\begin{center}
  \begin{tikzpicture}[inner sep=0pt]
\draw (4,3) node 		(racine) {$\bullet$};

\draw (1,2) node 		(a) {$\bullet$};
\draw (4,2) node 		(b) {$\bullet$};
\draw (7,2) node 		(c) {$\bullet$};

\draw (0,1) node 		(ab) {$\bullet$};
\draw (2,1) node 		(ac) {$\bullet$};
\draw (3,1) node 		(ba) {$\bullet$};
\draw (5,1) node 		(bc) {$\bullet$};
\draw (6,1) node 		(ca) {$\bullet$};
\draw (8,1) node 		(cb) {$\bullet$};

\draw (0,0) node 		(abc) {$\bullet$};
\draw (2,0) node 		(acb) {$\bullet$};
\draw (3,0) node 		(bac) {$\bullet$};
\draw (5,0) node 		(bca) {$\bullet$};
\draw (6,0) node 		(cab) {$\bullet$};
\draw (8,0) node 		(cba) {$\bullet$};

\path (racine) edge  node[left]{$\mathtt{io}_{a}\;\;\;\;\;$} (a);
\path (racine) edge  node[left]{$\mathtt{io}_{b}\;$} (b);
\path (racine) edge  node[right]{$\;\;\;\;\mathtt{io}_{c}$} (c);

\path (a) edge node[left]{$\mathtt{io}_{b}\;\;$} (ab);
\path (a) edge node[right]{$\;\;\mathtt{io}_{c}$} (ac);
\path (b) edge node[left]{$\mathtt{io}_{a}\;\;$} (ba);
\path (b) edge node[right]{$\;\;\mathtt{io}_{c}$} (bc);
\path (c) edge node[left,above]{$\mathtt{io}_{a}\;\;\;\;\;\;$} (ca);
\path (c) edge node[left,below]{$\mathtt{io}_{b}\;\;\,$} (cb);

\path (ab) edge node[left]{$\mathtt{io}_{c}\;$} (abc);
\path (ac) edge node[left]{$\mathtt{io}_{b}\;$} (acb);
\path (ba) edge node[left]{$\mathtt{io}_{c}\;$} (bac);
\path (bc) edge node[left]{$\mathtt{io}_{a}\;$} (bca);
\path (ca) edge node[left,pos=0.6]{$\mathtt{io}_{b}\;$} (cab);
\path (cb) edge node[left]{$\mathtt{io}_{a}\;$} (cba);

\draw[->,>=latex,color=blue, line width=1pt] (acb) to[bend right] (ac);
\draw[->,>=latex,color=blue, line width=1pt] (ba) to[bend right] (b);
\draw[->,>=latex,color=blue, line width=1pt] (bca) to[bend right] (bc);
\draw[->,>=latex,color=blue, line width=1pt] (ca) to[out=20,in=250] (c);
\draw[->,>=latex,color=blue, line width=1pt] (cb) to[bend right] (c);
\draw[->,>=latex,color=blue, line width=1pt] (cba) to[bend right] (cb);

\draw[-latex,color=red,dashed,line width=1pt] (cab) .. controls +(20:8mm) and +(290:4mm) .. (c);
\draw[-latex,color=red,dashed,line width=1pt] (cab) .. controls +(20:4mm) .. (ca);
\end{tikzpicture}
\end{center}

We use $\mathsf{io}_i$ as a shortcut for $\In(i,X_i)\cdot\Out(i,w_i)$ and
we represent  dependency constraints using arrows.
For instance, on the trace $\mathsf{io}_a\cdot\mathsf{io}_c\cdot\mathsf{io}_b$, a dependency constraint
of the form $\constrd{X_b}{w_c}$ (represented by the left-most arrow)
is generated. 
    Now, on the trace
    $\mathsf{io}_c\cdot\mathsf{io}_a\cdot\mathsf{io}_b$
    we add $\constrd{X_a}{w_c}$ after the second transition,
    and $\constrd{X_b}{\{w_c,w_a\}}$ (represented by the 
    dashed $2$-arrow) after the third transition. Intuitively,
    the latter constraint expresses that
    $\mathsf{io}_b$ is only allowed to come after $\mathsf{io}_c$ if it 
    depends on it, possibly indirectly through $\mathsf{io}_a$.
\end{example}

This reduced semantics gives rise to a notion of trace equivalence.
It is defined as in
Definition~\ref{def:equiv-comp-symb}, relying on $\ssymd{}$ instead of
$\ssymc{}$ and on $\Sol^2$ instead of $\Sol$.
We denote it $\esymd$, and the associated notion of
inclusion is denoted $\incesymd$

\subsection{Soundness and completeness}

In order to establish that $\esym$ and $\esymd$ coincide, we are going to
study more carefully concrete traces made of (not necessarily proper)
blocks. We denote by~$\Bio$ the set of blocks $\io{c}{M}{w}$ such that
$c\in\mathcal{C}$, $M_i\in\T(\W)$ for each $M_i\in\vect M$, and $w_j
\in \W$ for each $w_j \in \vect{w}$.
In this section, a concrete trace is necessarily made of blocks,
\ie it belongs to $\Bio^*$.
Note that all traces from executions in the compressed semantics
are concrete traces in this sense.
We show that we can view $\Bio^*$ as a partially commutative monoid
in a meaningful way.
This allows us to lift a classic result in which we ground our reduced
semantics.

  We lift the ordering on channels to blocks:
  $\io{c}{M}{w} \och \io{c'}{M'}{w'}$ if, and only if, $c \och c'$.
  Finally, we define $\och$ on concrete traces as the lexicographic
  extension of the order on blocks. We define similarly $\och$ on
  symbolic traces.

\paragraph*{Partially commutative monoid.}
  We define an \emph{independence relation} ${\I_b}$ over $\Bio$:
  we say that $\io{c}{M}{w} \;\I_b\; \io{c'}{M'}{w'}$ when
   $c \neq c'$,  none of the variables of $\vect w'$ occurs in
  $\vect M$, and none of the variables of $\vect w$ occurs in $\vect M'$.
  Then we define ${=}_{\I_b}$ as the least congruence satisfying \\[1mm]
  \null\hfill $
\io{c}{M}{w}\cdot \io{c'}{M'}{w'} =_{\I_b} \io{c'}{M'}{w'}\cdot
\io{c}{M}{w}
  $\hfill\null\smallskip \\ \noindent
for all
  $\io{c}{M}{w}$ and $\io{c'}{M'}{w'}$ with $\io{c}{M}{w} \;\I_b\; \io{c'}{M'}{w'}$.
  The set of concrete traces, quotiented by this
  equivalence relation, is the \emph{partially commutative monoid}
  obtained from $\I_b$.
%
  Given a concrete trace $\tr$, we denote 
by $\mini(\tr)$ the minimum
  for $\och$ among all the traces that are equal to $\tr$ modulo $=_{\I_b}$.
\smallskip{}

  First, we prove that the symbolic semantics is equally able
    to execute equivalent (w.r.t. $=_{\I_b}$) traces.
    Second we prove that the reduced semantics generates
    dependency constraints that are (only) satisfied
    by minimal traces.

\begin{restatable}{lemma}{lemmadiffswap}
  \label{lem:diff-swap}
  Let $\triple{\p_0}{\Phi_0}{\varnothing} \ssymc{\tr}
  \triple{\p}{\Phi}{\Set}$ with $\tr$ made of proper blocks,
  and $\theta\in\Sol\cs{\Phi}{\Set}$.  For any concrete trace
  $\tr_c =_{\I_b} \tr\theta$ there exists a symbolic trace $\tr'$
  such that $\tr_c = \tr'\theta$,
  $\triple{\p_0}{\Phi_0}{\varnothing} \ssymc{\tr'}
  \triple{\p}{\Phi}{\Set'}$ and $\theta\in\Sol(\Phi;\Set')$.
\end{restatable}

\begin{restatable}{lemma}{lemmadiffmin}
  \label{lem:diff-min}
  Let $A \ssymc{\tr} \triple{\p}{\Phi}{\Set}$
  and $\theta\in\Sol\cs{\Phi}{\Set}$.
  We have that $\theta \models \AllDep{\tr}$ if, and only if,
  $\tr\theta = \mini(\tr\theta)$.
\end{restatable}

\begin{proof}[Sketch]
    Let $A \ssymc{\tr} \triple{\p}{\Phi}{\Set}$
  and $\theta\in\Sol\cs{\Phi}{\Set}$.
  We need a characterization  of minimal traces. We exploit
  the following one, which is equivalent to the characterization of
  Anisimov and Knuth~\cite{knuth-sorting}:
\unskip{}
\begin{quote}
 The trace $t$ is minimal if, and only if,
    for all factors $aub$ of $t$ such that
    (1) $a,b\in \Bio$, $u\in \Bio^*$ and
    $d \och b \och a$ for all $d \in u$,
    we have (2) some $c\in a u$ such that
    $c \;\I_b\; b$ does not hold.
\end{quote}
We remark that condition (1) characterizes the factors of (symbolic) traces
for which we generate a dependency constraint. Here, that constraint would be 
\\[1mm]
\null \hfill $\constrd{\vect X_b}{\cup_{d\in au}{\vect w_d}}$
\hfill\null
\smallskip{}

\noindent where $\alpha \in \Bio$ is also written
$\mathsf{io}_{c_\alpha}(\vect X_\alpha,\vect w_\alpha)$ to have an
access to its components. 

Then we note that (2) corresponds to the satisfaction of that dependency
constraint in a concrete instance of the trace.
\qed
\end{proof}

Finally, relying on these results, we can establish  
that trace equivalence ($\approx_d$) w.r.t. the reduced semantics exactly captures
symbolic trace equivalence ($\approx_s$). 

\begin{restatable}{theorem}{theoesymesymd}
\label{theo:esymesymd}
For any extended simple processes $A$ and $B$, we have that:\\[1mm]
\null\hfill
$A \incesym B \iff A \incesymd B$.
\hfill\null
\end{restatable}

\begin{proof}[Sketch]
  Implication $(\Rightarrow)$ is straightforward and only
relies on the fact that dependency constraints generated
by the reduced semantics only depend on the trace that is
executed.
The other direction $(\Leftarrow)$ is more interesting.
Here, we only outline the main idea, in the case of a trace made of proper blocks.
We show that a concrete trace $\tr\theta$ which is not captured when considering
$\ssymd{}$ (\ie a trace $\tr\theta$ that does not satisfy the generated dependency
constraints) can be mapped to another trace, namely $\mini{(\tr\theta)}$,
which manipulates the same
recipes/messages but where blocks are executed in a different order.
Lemma~\ref{lem:diff-swap} is used to obtain an execution of the minimal trace,
and
Lemma~\ref{lem:diff-min} ensures that dependency constraints
are satisfied in that execution. Thus the minimal trace can also be
executed by the other process. We go back to $\tr\theta$ using
Lemma~\ref{lem:diff-swap}.
\qed
\end{proof}

\subsection{First-order reduced semantics}

We finally introduce the stronger, first-order semantics for dependency
constraints, and we prove soundness and completeness for the corresponding
equivalence property by building on the previous theorem.


\begin{definition}
  Let $\C = \cs{\Phi}{\Set}$ be a dependency constraint system.
  We define $\Sol^1(\C)$ to be the set of substitutions
  $\theta\in\Sol(\noconstrd{\C})$ such that, for each
  $\constrd{\vect X}{\vect w}$ in $\C$ with non-empty $\vect w$
  there is some $X_i\in\vect{X}$ such that
  for all recipes $M\in\T(D_\C(X))$ satisfying
  $M(\Phi\lambda_\theta) {=_{\E}} (X\theta)(\Phi\lambda_\theta)$,
  we have $\fvp(M) \cap \vect w \neq \varnothing$.
\end{definition}

We define the notion of trace equivalence accordingly,
as it has been done at the end of Section~\ref{subsec:reduced-semantics},
relying on $\Sol^1$ instead of $\Sol^2$.
We denote it $\esymdp$, and the associated notion of
inclusion is denoted $\incesymdp$.



\begin{restatable}{theorem}{theodifffirst}
\label{theo:diff-first}
For any extended simple processes $A$ and $B$, we have that:\\[1mm]
\null\hfill$A \incesymd B \iff A \incesymdf B$.\hfill\null
\end{restatable}

\begin{proof}[Sketch]
($\Rightarrow$) This implication is relatively easy to establish. It
actually relies on the fact that $\Sol^1(\C) \subseteq
\Sol^2(\C)$ for any dependency constraint system~$\C$. This allows us to use our hypothesis $A \incesymd B$. 
Then, in order to come back to our more constrainted first-order reduced semantics, 
we may notice that as soon as~$\theta$ is a solution of $\C$ and
$\C'$ (w.r.t. $\Sol^2$) with static equivalence of their associated
frames, we have that: $\theta \in \Sol^1(\C)$ if, and only if, $\theta \in \Sol^1(\C')$.
%
($\Leftarrow$) In order to exploit our hypothesis $A \incesymdf B$, given a trace
\\[1mm] \null\hfill $A \ssymd{\tr}
\triple{\p}{\Phi}{\Set}$ with $\theta \in \Sol^2(\Phi; \Set)$, 
\hfill\smallskip \\
we build
 $\tr'$ and $\theta'$ such that $A \ssymd{\tr'}
\triple{\p}{\Phi}{\Set'}$ with ``$\tr =_{\Ind_b} \tr'$'', and  $\theta'
\in \Sol^1(\Phi; \Set')$. Actually, we do this without changing the underlying
first-order substitution, \ie $\lambda_\theta =
\lambda_{\theta'}$.
This is done by a sub-induction; iteratively modifying $\theta$ and $\tr$.
Whenever $\theta$ is not already a first-order solution, we slightly modify
it. We obtain a new substitution $\theta'$ that is not a second-order solution
anymore w.r.t. $\tr$, and we use
Lemmas~\ref{lem:diff-swap}~and~\ref{lem:diff-min}
to obtain a new trace $\tr' \och \tr$ for
which $\theta'$ is a second-order solution.
By induction hypothesis on $\tr'$ we obtain a first-order
solution. We finally go back to the original trace $\tr$, using an argument
similar to the one in the first direction to handle static equivalence.
\qed 
\end{proof}

\begin{example}
\label{ex:iteration}
We illustrate the construction of $\tr'$, which is at the core
of the above proof.
Consider $A = \proc{\{P_1, P_2, P_3\}}{\Phi}$ where
$P_i=\In(c_i,x_i).\Out(c_i,n_i)$, and $\Phi_0 = \{w_0 \refer
n_0\}$, and $n_i \in \N$ for $0\leq i\leq 3$. We assume that $c_1\och c_2\och c_3$, and we consider the
situation where the nonces
$n_0$ and $n_2$ (resp. $n_1$ and $n_3$) are the same.

Let 
$\tr = \iossvect{c_3}{X_3}{w_3}.
\iossvect{c_2}{X_2}{w_2}.
\iossvect{c_1}{X_1}{w_1}$
and $\cs{\Phi}{\Set}$ the dependency constraint
system such that $A \ssymd{\tr} \triple{\varnothing}{\Phi}{\Set}$.
We consider the substitution
$\theta = \{X_3\mapsto \start,\, X_2\mapsto w_3, \,X_1\mapsto w_2\}$.
We note that $\theta \in \Sol^2(\Phi;\Set)$ but we have that $\theta \not\in
\Sol^1(\Phi;\Set)$ due to the presence of 
$\constrd{X_1}{w_2}$ in $\Set$.
We could try to fix this problem by building a ``better'' solution
$\theta'$ 
that yields the same first-order solution:
$\theta'=\{X_3\mapsto \start, \, X_2 \mapsto w_3,\,X_1\mapsto w_0\}$
is such a candidate.
Applying Lemmas~\ref{lem:diff-swap}~and~\ref{lem:diff-min}, we obtain 
a smaller symbolic trace:
\\[1mm] \null \hfill $
\tr' = \iossvect{c_1}{X_1}{w_1} \cdot
\iossvect{c_3}{X_3}{w_3} \cdot
\iossvect{c_2}{X_2}{w_2}
$. \hfill \smallskip \\
Let $\cs{\Phi}{\Set'}$ be the constraint system obtained from the
execution of $\tr'$.
We have that $\theta' \in \Sol^2(\Phi;\Set')$ but again $\theta' \not
\in \Sol^1(\Phi;\Set')$. This is due to the presence of
$\constrd{X_2}{w_3}$ in $\Set'$ --- which was initially satisfied by
$\theta$ in the first-order sense.
With one more iteration of this transformation, we obtain a third
candidate:
$\theta''=\{X_3\mapsto \start,\,X_2\mapsto  w_1,\,X_1\mapsto w_0\}$ and
\\[1mm]\null\hfill $\tr'' = \iossvect{c_1}{X_1}{w_1} \cdot
\iossvect{c_2}{X_2}{w_2} \cdot
\iossvect{c_3}{X_3}{w_3}$. \hfill\smallskip \\
The associated constraint system does not contain any
dependency constraint, and thus
$\theta''$ is trivially a first-order solution.
\end{example}
\subsection{Applications}

We first describe two situations 
showing that  our reduced semantics can yield an exponential benefit.
Then, we illustrate the effect of our reduced semantics on our
running example, \ie the private authentication protocol. 

\smallskip{}

  Consider first the simple process
  $\mathcal{P} = \{P_1, P_2,\ldots, P_n\}$ where
  each $P_i$ denotes the basic process
  $\In(c_i,x).\testt{x = \ok}{\Out(c_i,n_i)}$ with $n_i \in \N$. 
There are $(2n)!/2^n$
  different traces  of size $2n$
  (\ie containing $2n$ visible actions) in the concrete
  semantics. This number is actually
  the same in the standard symbolic semantics.
  In the compressed semantics (as well as the symbolic
  compressed semantics) this number goes down to $n!$.
  Finally, in the reduced semantics, there is only one
  trace such that the resulting constraint system admits a
  solution. Assuming that $c_1 \och \ldots \och c_n$, that trace
  is simply:\\[2mm]
\null\hfill
$\tr = \io{c_1}{X_1}{w_1} \cdot \ldots \cdot\io{c_n}{X_n}{w_n}.$\hfill\null


\medskip{}

Next, we consider the simple process
  $\mathcal{P} = \{P_1^n, \, P_2^n\}$ where  $P_i^0 = 0$, and 
  $P_i^{n+1}$ denotes the basic process
  $\In(c_i,x_j).\testt{x_j = \ok}{\Out(c_i,n_j)}. P_i^{n}$.
  We consider traces of size~$4 n$.
  In the concrete semantics, there are $\binom{4n}{2n}$
  different traces, whereas the number of such  traces is reduced to
  $\binom{2n}{n}$ in the compressed semantics.
  Again, there is only one trace left in the reduced semantics.


\medskip{}

Going back to our running example (see Examples~\ref{ex:private}
 and~\ref{ex:initial}),  
we represent some symbolic traces obtained using our reduced
semantics. We 
 consider:\\[1mm]
\null\hfill
$\triple{\{P_\init,
  Q_0(\skb,\pk{\ska})\}}{\Phi_0}{\varnothing}$\hfill\null
\smallskip{}

\noindent  and we assume that $c_A \och c_B$.
 We consider all symbolic traces obtained without considering the 
 \textsc{Else} rule.

\setlength{\columnsep}{10pt}%
\begin{wrapfigure}[8]l{4cm}
   \begin{tikzpicture}[inner sep=0pt]
     \draw (2,3) node 		(racine) {$\bullet$};

     \draw (1,2) node 		(a1) {$\bullet$};
     \draw (3,2) node 		(b) {$\bullet$};

     \draw (0,1) node 		(a1a2) {$\bullet$};
     \draw (2,1) node 		(a1b) {$\bullet$};
     \draw (3,1) node              (ba1) {$\bullet$};

     \draw (2,0) node 		(a1ba2) {$\bullet$};
     \draw (3,0) node 		(ba1a2) {$\bullet$};

     \path (racine) edge  node[left]{$\bio^1_a\;\;$} (a1);
     \path (racine) edge  node[right]{$\;\;\bio_b$} (b);

     \path (a1) edge node[left]{$\bio^2_{a}\;\;\;$} (a1a2);
     \path (a1) edge node[left,below]{$\bio_{b}\;\;\;\;\;\;$} (a1b);
     \path (b) edge node[right]{$\bio^1_{a}$} (ba1);

     \path (a1b) edge node[left]{$\bio^2_{a}$} (a1ba2);
     \path (ba1) edge node[right]{$\bio^2_{a}$} (ba1a2);


     \draw[->,>=latex,color=blue, line width=1pt] (ba1) to[bend left] (b);
     \draw[->,>=latex,color=blue, line width=1pt] (a1ba2) to[bend right] (a1b);
   \end{tikzpicture}
 \end{wrapfigure}

Those executions are represented in the diagram on the 
 left, where\vspace{1pt}
   \begin{itemize}
   \item  $\bio^1_a$ to denote 
     $\iossvect{c_A}{X^1_a}{w_a}$,
   \item $\bio^2_a$ to denote $\iossvect{c_A}{X^2_a}{\varnothing}$, and
\item $\bio_b$ to denote  $\iossvect{c_B}{X_b}{w_b}$.
\end{itemize}
The block $\bio^2_a$ is an improper block since it only contains an input
action. 
First, we may note that many interleavings are not taken into account
thanks to compression. 
Now, consider the symbolic trace $\bio^1_a \cdot \bio_b \cdot \bio^2_a$.
A dependency constraint of the form $\constrd{X^2_a}{w_b}$ is
generated. Thus, a concrete trace that satisfies this dependency
constraint
must use the output of the role $Q_0(b,a)$ to build
the second input of the role $P_\init$.

Second, consider the rightmost branch. A dependency
constraint of the form $\constrd{X^1_a}{w_b}$ is generated, and since
$X^1_a$ has to be instantiated by a recipe that gives the public constant~$\start$
(due of the constraint $x^1_a =^{?} \start$ present in the system),
the reduced semantics makes it possible to prune all executions starting with
$\bio_b \cdot \bio_a^1$.




%% file: conclu.tex
\section{Conclusion}
\label{sec:conclu}

We have proposed two refinements of the symbolic semantics for simple
processes. 
The first refinement groups actions in blocks, while
the second one uses dependency constraints to restrict to minimal 
interleavings among a class of permutations. In both cases, the
refined semantics has less traces, yet we show that the associated
trace equivalence coincides with the standard one. In theory, this
yields a potentially exponential algorithmic optimization.


In order to validate our approach, an experimental implementation has
been developed\footnote{
  Available at 
  \url{<http://perso.ens-lyon.fr/lucca.hirschi/spec_en.html>}.
}.
This tool is based on SPEC~\cite{SPEC}
(which does not support $\mathsf{else}$ branches)
and implements our modified semantics as well as an adapted constraint
resolution procedure that takes (first-order) dependency constraints
into account. The latter procedure is quite preliminary and far from
optimal. Yet, the modified checker already shows significant improvements
over the original version on various benchmarks (\cite{rr-lsv-13-13}, Figure~9).

\smallskip{}

We are considering several directions for future work.
Constraint solving procedures should be studied in depth:
we may optimize the one we already developed~\cite{rr-lsv-13-13} and we are
also interested in studying the problem in other frameworks,
\eg\cite{APTE}.
We also believe that stronger reductions can be achieved:
for instance, exploiting symmetries should be very useful for
dealing with multiple sessions.



%% file: app-compression.tex
\section{Proofs of Section~\ref{sec:compression}}

\renewcommand{\LRstep}[1]{\xRightarrow{#1}}



\lempermute*
\begin{proof}
  It suffices to establish that
  $A \LRstep{\alpha\cdot\alpha'\;} A'$ implies
  $A \LRstep{\alpha'\cdot\alpha\;} A'$
  for any $\alpha\;\I_{a}\;\alpha'$.
\begin{itemize}
\item
  \newcommand{\acti}{\Out(c_i,w_i)}
  \newcommand{\actj}{\Out(c_j,w_j)}
  Assume that we have
  $A \LRstep{\acti} A_i \LRstep{\actj} A'$ with $c_i \neq c_j$.
  Because we are considering simple processes,
  the two actions must be concurrent.
  More specifically, our process $A$ must be of the form
  $\proc{\{P_i, P_j\}\uplus\p_r}{\Phi}$ with $P_i$ (resp.~$P_j$)
  being a basic process on channel $c_i$ (resp.~$c_j$).
  We assume that in our sequence of reductions, $\tau$ actions
  pertaining to $P_i$ are all executed before reaching $A_i$,
  and that $\tau$ actions pertaining to $\p_r$ are executed last.
  This is without loss of generality, because a $\tau$ action
  on a given basic process can easily be permuted with actions
  taking place on another basic process, since it does not depend
  on the context and has no effect on the frame.
  Thus we have that
  $A_i = \proc{\{P'_i,P_j\}\uplus\p_r}{\Phi\uplus\{w_i\refer m_i\}}$,
  $A'  = \proc{\{P'_i;P'_j\}\uplus\p'_r}{\Phi\uplus\{w_i\refer m_i, w_j\refer 
  m_j\}}$. Since the $\tau$ actions taking place on $\p'_r$ rely neither
  on the frame nor on the first two basic processes, we easily obtain
  the permuted execution:
  \[ \begin{array}{rcl}
    A
    & \LRstep{\actj} &
    \proc{\{P_i,P'_j\}\uplus\p_r}{\Phi\uplus\{w_j\refer m_j\}} \\
    & \LRstep{\acti} &
    \proc{\{P'_i,P'_j\}\uplus\p'_r}{\Phi\uplus\{w_i\refer m_i,w_j \refer 
    m_j\}}
  \end{array} \]

\item
  The permutation of two input actions on distinct channels is very
  similar. In this case, the frame does not change at all, and the
  order in which messages are derived from the frame does not matter.
  Moreover, the instantiation of the input variable on one basic
  process has no impact on the other ones.

\item
  \renewcommand{\acti}{\Out(c_i,w_i)}
  \renewcommand{\actj}{\In(c_j,R)}
  Assume that we have
  $A \LRstep{\acti} A_i \LRstep{\actj} A'$ with $c_i \neq c_j$ and 
  $w_i\not\in \fv(R)$.
  Again, the two actions are concurrent, and we can assume that
  $\tau$ actions are organized conveniently so that
  $A$ is of the form $\proc{\{P_i,P_j\}\uplus\p_r}{\Phi}$
  with $P_i$ (resp. $P_j$) a basic process on $c_i$ (resp. $c_j$);
  $A_i$ is of the form $\proc{\{P'_i,P_j\}\uplus\p_r}{\Phi\uplus\{w_i\refer 
  m_i\}}$;
  and $A'$ is of the form
  $\proc{\{P'_i,P'_j\}\uplus\p'_r}{\Phi\uplus\{w_i\refer m_i\}}$.
  As before, the $\tau$ actions from $\p_r$ to $\p'_r$ are easily moved 
  around. Additionnally, $w_i \not\in \fv(R)$ implies
  $\fv(R)\subseteq\dom(\Phi)$ and thus we have:
  \[
    \proc{\{P_i,P_j\}\uplus\p_r}{\Phi}
    \LRstep{\actj} \proc{\{P_i,P'_j\}\uplus\p_r}{\Phi} \]
  The next step is trivial:
  \[ \proc{\{P_i,P'_j\}\uplus\p_r}{\Phi}
     \LRstep{\acti}
   \proc{\{P'_i,P'_j\}\uplus\p'_r}{\Phi\uplus\{w_i\refer m_i\}} \]

\item
  We also have to perform the reverse permutation, but we shall
  not detail it; this time we are delaying the derivation of $R$ from the 
  frame, and it only gets easier.

\qed
\end{itemize}
\end{proof}


\procompsoundness*

\begin{proof}
  This results immediately follows from
  the observation that $\sintc{}$ is included in $\LRstep{}$
  for traces made of proper blocks since for them
  {\sc Failure} can not arise.
\end{proof}




\procompcompleteness*

\begin{proof}
  We first observe that
  $A \LRstep{\tr} A'$ implies $A \sintf{\tr}{o^*} A'$
  if $A'$ is initial and $\tr$ is a (possibly empty)
  sequence of output actions on the same channel.
  We prove this by induction on the sequence of actions.
  If it is empty, we can conclude using
  one of the \textsc{Proper} rules because $A = A'$ is initial.
  Otherwise, we have:
  $$A \LRstep{\Out(c,w)} A'' \LRstep{\tr\;} A'.$$
  We obtain $A'' \sintf{\tr}{o^*} A'$ by induction hypothesis,
  and conclude using rules \textsc{Tau} and \textsc{Out}.

  The next step is to show that
  $A \LRstep{\tr\;} A'$ implies $A \sintf{\tr}{i^{*}} A'$,
  if $A'$ is initial and $\tr$ is the
  concatenation of a (possibly empty) sequence of inputs
  and a non-empty sequence of outputs,
  all on the same channel.
  This is easily shown by induction on the number of input actions.
  If there is none we use the previous result, otherwise we conclude
  by induction hypothesis and using rules \textsc{Tau} and \textsc{In}.
  Otherwise, the first output action allows us to conclude from the
  previous result and rules \textsc{Tau} and \textsc{Out}.

  We can now show that
  $A \LRstep{\tr\;} A'$ implies $A \sintf{\tr}{i^{+}} A'$
  if $A'$ is initial and $\tr$ is a proper block.
  Indeed, we must have $$A \LRstep{\In(c,R)} A'' \LRstep{\tr'\;} A'$$
  which allows us to conclude using the previous result and rules
  \textsc{Tau} and \textsc{In}.

  We finally obtain that $A \LRstep{\tr} A'$ implies
  $A \sintc{\tr} A'$ when $A$ and $A'$ are initial \emph{simple} processes
  and $\tr$ is a sequence of proper blocks. This is done by induction
  on the number of blocks. The base case is trivial.
  Because $A$ is initial, the execution of its basic processes can only
  start with visible actions, thus only one basic process is involved in
  the execution of the first block.
  Moreover, we can assume without loss of generality that the execution of
  this first block results in another initial process: indeed the basic process
  resulting from that execution is either in the final process $A'$, which is
  initial, or it will perform another block, \ie it can perform
  $\tau$ actions followed by an input, in which case we can force
  those $\tau$ actions to take place as early as possible. Thus we have
  \[ A \LRstep{\mathsf{b}} A'' \LRstep{\tr'} A' \]
  where $\mathsf{b}$ is a proper block, and we conclude using the previous
  result and the induction hypothesis.
\qed
\end{proof}


\begin{proposition} \label{prop:focus-trace}
  Let $\tr$ be a trace of visible actions such that,
  for any channel $c$ occurring in the trace, it appears first
  in an input action.
  There exists a sequence of proper blocks $\tr_{io}$ and
  a sequence of improper blocks\footnote{
    An improper block is composed of
    a sequence of input actions on the same channel.
  } $\tr_i$
  such that $\tr =_{\I_a} \tr_{io}\cdot\tr_i$.
\end{proposition}

\begin{proof}
  We proceed by induction on the length of $\tr$, and distinguish
  two cases:
\begin{itemize}
\item If $\tr$ has no output action then, by swapping input actions
  on distinct channels,
  we reorder $\tr$ so as to obtain
  $\tr_i = \tr^{c_1}\cdot\ldots\cdot\tr^{c_n}=_{\I_a}\tr$
  where the $c_i$'s are pairwise distinct and
  $\tr^{c_i}$ is an improper block on channel $c_i$.
\item Otherwise,
  there must be a decomposition $\tr = \tr_1\cdot\out{c}{w}\cdot\tr_2$
  such that
  $\tr_1$ does not contain any output. We can perform swaps involving
  input actions of $\tr_1$ on all channel $c'\neq c$, so that they
  are delayed after the first output on $c$.
  We obtain $\tr =_{\I_a}
  \inp{c}{M_1}\cdot\ldots\cdot\inp{c}{M_n}\cdot\out{c}{w}
  \cdot \tr_1'\cdot \tr_2$
  with $n\geq 1$.
  Next,
  we swap output actions on channel $c$ from $\tr_1'\cdot\tr_2$ that
  are not preceded by another input on $c$, so as to obtain
  \[ \tr =_{\I_a}
  \inp{c}{M_1}\dots\inp{c}{M_n}\cdot\out{c}{w}
  \cdot\out{c}{w_1}\dots\out{c}{w_m}\cdot\tr_2' \] such that
  either $\tr_2'$ does not contain any action on channel $c$
  or the first one is an input action.
  We have thus isolated a first proper block,
  and we can conclude by induction hypothesis on $\tr_2'$.
  \qed
\end{itemize}
\end{proof}

Note that the above result does not exploit all the richness of $\I_a$.
In particular, it never relies on the possibility to swap an input action
before an output when the input message does not use the output handled.
Indeed, the idea behind compression does not rely on messages.
This will change in Section~\ref{sec:diff} where we will use ${\I_a}$
more fully.

\smallskip


We finally prove the main result about the compressed semantics.
Given two simple process $A = \proc{\p}{\Phi}$ and $A' =
\proc{\p'}{\Phi'}$, we shall write $\Phi(A) \statequiv \Phi(A')$ (or
even $A \statequiv A'$) instead of $\Phi
  \statequiv \Phi'$.

\theocompsoundnesscompleteness*

\begin{proof} We prove the two directions separately.

\noindent $(\Rightarrow)$ 
Let $A$ be an initial simple process such that $A \approx B$
and ${A \sintc{\tr} A'}$.
%
One can easily see that our trace $\tr$ must be of the form
$\tr_{io} \cdot \tr_i$ where $\tr_{io}$ is made of proper blocks
and $\tr_i$ is a (possibly empty) sequence of inputs
on the same channel $c_j$.
We have: $$A \sintc{\tr_{io}} A'' \sintc{\tr_i} A'$$

Using Proposition~\ref{pro:comp-soundness}, we obtain that
$A \LRstep{\tr_{io}} A''$.
We also claim that $A''\LRstep{\tr_i} A^{+}$
for some $A^{+}$ having the same frame as $A'$.
This is obvious when $\tr_i$ is empty ---
in that case we can simply choose $A^{+} = A' = A''$.
Otherwise, the execution of the improper block $\tr_i$ results from the 
application of rule \textsc{Improper}. Except for the fact that this rule
``kills'' the resulting process, its subderivation simply packages a
sequence of inputs, and so we have a suitable $A^+$.
We thus have:
\[
A\LRstep{\tr_{io}\;}
A''\LRstep{\tr_i\;}
A^{+} \]

By hypothesis, it implies that
$B \LRstep{\tr_{io}\;} B''$
and $B\LRstep{\tr_{io}\cdot\tr_i\;} B^+$
with $A'' \statequiv B''$ and $A^+ \statequiv B^+$.
Since simple processes are determinate, we have:
\[ B \LRstep{\tr_{io}} B'' \LRstep{\tr_i} B^+ \]

It remains to establish that $B \sintc{\tr} B'$
such that $B'\estat A'$.
We can assume that $B''$ does not have any basic process
starting with a test, without loss of generality since forcing
$\tau$ actions cannot break static equivalence.
Further, we observe that $B''$ is initial. Otherwise, it
would mean that a basic process of $B$ is not initial (absurd)
or that one of the blocks of $\tr_{io}$, which are maximal for
$A$, is not maximal for $B$ (absurd again, because it contradicts
$A \approx B$).
This allows us to apply Proposition~\ref{pro:comp-completeness}
to obtain $$B \sintc{\tr_{io}} B''.$$
This concludes when $\tr_i$ is empty, because $B' = B'' \statequiv A'' = A'$.
Otherwise, we note that $A^+$ cannot perform any action on channel $c_j$,
because the execution of $\tr_i$ in the compressed semantics must be maximal.
Since $A \approx B$, it
must be that $B^+$ cannot perform any visible action on the channel $c_j$
either. Thus $B''$ can complete an improper step:
\[ B''\sintc{\tr_i}B' \mbox{ where } B' = \proc{\varnothing}{\Phi(B^+)}. \]
We can finally conclude that
$B \sintc{\tr} B'$ with $\Phi(B') = \Phi(B^+) \statequiv \Phi(A^+) = \Phi(A')$.

\bigskip{}

\noindent $(\Leftarrow)$
Let $A$ be an initial simple process such that $A \approxc B$
and $A \LRstep{\tr} A'$. We ``complete'' this execution as follows:
\begin{itemize}
\item We force $\tau$ actions whenever possible.
\item If the last action on $c$ in $\tr$ is an input, we trigger
  available inputs on $c$ using an arbitrary public constant as a
  recipe.
\item We trigger all the outputs that are available.
\end{itemize}

We obtain a trace of the form $\tr\cdot\tr^{+}$.
Let $A^{+}$ be the process obtained from this trace:
\[ A \LRstep{\tr} A' \LRstep{\tr^{+}} A^{+} \]
We observe that $A^{+}$ is initial: indeed, for each basic process
that performs actions in $\tr\cdot\tr^{+}$, either the last action on its
channel is an output and the basic process is of the form
$\In(c,\_).P$, or the last action is an input and the basic
process is reduced to $0$ and disappears.

Next, we apply Proposition~\ref{prop:focus-trace} to obtain traces
$\tr_{io}$ (resp.~$\tr_i$) made of proper (resp.~improper) blocks,
such that $\tr\cdot\tr^{+} =_{\I_{a}} \tr_{io}\cdot\tr_i$.
By Lemma~\ref{lem:permute} we know that this permuted trace
can also lead to $A^{+}$:
\[
  A \LRstep{\tr_{io}} A_{io} \LRstep{\tr_i} A^{+}
\]
As before, we can assume that $A_{io}$ cannot perform any $\tau$
action. Under this condition, since $A^{+}$ is initial,
$A_{io}$ must also be initial.

By Proposition~\ref{pro:comp-completeness} we have that
$A \sintc{\tr_{io}} A_{io}$,
and $A \approxc B$ implies that:
$$B \sintc{\tr_{io}} B_{io} \mbox{ with } \Phi(A_{io}) \statequiv \Phi(B_{io}). $$
A simple inspection of the \textsc{Proper} rules shows that a basic
process resulting from the execution of a proper block must be
initial. Thus, since the whole simple process $B$ is initial,
{$B_{io}$ is initial too}.

Thanks to Proposition~\ref{pro:comp-soundness}, we have that
$B \LRstep{\tr_{io}} B_{io}$.
Our goal is now to prove that we can complete this execution with $\tr_i$.
This trace is of the form $\tr^{c_1}\cdot\tr^{c_2}\dots\tr^{c_n}$ where $\tr^{c_i}$ contains
only inputs on channel $c_i$ and the $c_i$ are pairwise disjoint.
Now, we easily see that for each $i$,
$$A_{io} \LRstep{\tr^{c_i}} A_i$$
and $A_i$ has no more atomic process on channel $c_i$.
Thus we have $A_{io} \sintc{\tr^{c_i}} A_i^0$ with
$A_i^0 = \proc{\varnothing}{\Phi(A_i)}$. Since
$A \approx_c B$, we must have some $B_i^0$ such that:
$$B \sintc{\tr_{io}} B_{io} \sintc{\tr^{c_i}} B_i^0 $$
We can translate this back to the regular semantics, obtaining
$
B \LRstep{\tr_{io}} B_{io} \LRstep{\tr^{c_i}} B_i
$.
We can now execute all these
inputs to obtain an execution of $\tr_{io}\cdot\tr_i$ towards some process
$B^{+}$:
$$B \LRstep{\tr_{io}} B_{io} \LRstep{\tr_i} B^{+}$$
Permuting those actions, we obtain thanks to Lemma~\ref{lem:permute}:
\[ B \LRstep{\tr} B' \LRstep{\tr^{+}} B^{+} \]
We observe that
$\Phi(B^{+}) = \Phi(B_{io}) \statequiv \Phi(A_{io}) \statequiv \Phi(A^{+})$,
and it immediately follows that $A' \statequiv B'$ because those frames
have the same domain, which is a subset of that of $\Phi(A^{+}) \statequiv \Phi(B^{+})$.
\qed
\end{proof}

%% file: app-symbolique.tex
\section{Appendix of Section~\ref{sec:constraint-solving}}
\label{sec:app-symbolic}

\subsection{Focused and compressed symbolic semantics}
\label{subsec:app-symbolic-semantics}

\noindent {\bf Focused symbolic semantics.} The main idea of
the compressed symbolic semantics is to ensure that when a process
starts executing some actions, it executes a maximal block of actions.

\smallskip{}

\noindent \resizebox{\textwidth}{!}{
  $
  \begin{array}{lc}
\mbox{\sc In} &  \infer[\mbox{with }\ell\in\{i^*;i^+\}]
  {\triple{P}{\Phi}{\Set}\ssymf{\In(c,X). \tr\;}{\ell} \triple{P''}{\Phi''}{\Set''}}
  {\triple{P}{\Phi}{\Set}\ssym{\In(c,X)} \triple{P'}{\Phi'}{\Set'}
  & \triple{P'}{\Phi'}{\Set'}\ssymf{\;\tr\;\;}{i^*} \triple{P''}{\Phi''}{\Set''}}
   \\ [4mm]
\mbox{\sc Out} & 
 \infer[\mbox{with } \ell\in\{i^*;o^*\}]
   {\triple{P}{\Phi}{\Set}\ssymf{\Out(c,w). \tr\;}{\ell} 
   \triple{P''}{\Phi''}{\Set''}}
   {\triple{P}{\Phi}{\Set}\ssym{\Out(c,w)} \triple{P'}{\Phi'}{\Set'}
   & \triple{P'}{\Phi'}{\Set'}\ssymf{\;\;\tr\;\;}{o^*} \triple{P''}{\Phi''}{\Set''}}
 \\ [4mm]
\mbox{\sc Tau} &
 \infer[\mbox{with }\ell\in\{o^*;i^+;i^*\}]
 {\triple{P}{\Phi}{\Set} \ssymf{\;\;\tr\;\;}{\ell} \triple{P''}{\Phi''}{\Set''}}
 {\triple{P}{\Phi}{\Set} \ssym{\tau} \triple{P'}{\Phi'}{\Set'} &
    \triple{P'}{\Phi'}{\Set'} \ssymf{\;\;\tr\;\;}{\ell} \triple{P''}{\Phi''}{\Set''}}
 \\ [4mm]
\mbox{\sc Proper} &
\infer{
   {\triple{0}{\Phi}{\Set}\ssymf{\;\;\epsilon\;\;}{o^*}
     \triple{0}{\Phi}{\Set}}
   }{}
\quad
\infer{
   {\triple{\In(c,x).P}{\Phi}{\Set}\ssymf{\;\;\epsilon\;\;}{o^*}
     \triple{\In(c,x).P}{\Phi}{\Set}}
   }{}
\\ [4mm]
\mbox{\sc Improper} &
\infer{
   {\triple{0}{\Phi}{\Set}\ssymf{\;\;\epsilon\;\;}{i^*} \triple{\bot}{\Phi}{\Set}}
 }{}
\end{array}
$
}

\bigskip{}

\noindent {\bf Compressed symbolic semantics.}
We define the compressed symbolic reduction $\ssymc{}$ between
symbolic processes as the least reflexive transitive relation
satisfying the following rules:

\smallskip{}

\noindent \resizebox{\textwidth}{!}{
$
\begin{array}{lcl}
  \infer
  {\triple{\{Q\}\uplus \p}{\Phi}{\Set}\ssymc{\;\;\tr\;\;}
    \triple{\{Q'\}\uplus\p}{\Phi'}{\Set'}}
  {\triple{Q}{\Phi}{\Set}\ssymf{\;\;\tr\;\;}{i^+} \triple{Q'}{\Phi'}{\Set'}
    & Q'\neq\bot} 
&\;\;\;\;\;&
  \infer
  {\triple{\{Q\}\uplus\p}{\Phi}{\Set}\ssymc{\;\;\tr\;\;} 
  \triple{\varnothing}{\Phi'}{\Set'}}
  {\triple{Q}{\Phi}{\Set}\ssymf{\;\;\tr\;\;}{i^+} \triple{Q'}{\Phi'}{\Set'} & Q' = \bot}
  \end{array}
$}

\subsection{Soundness and completeness}

\begin{proposition}
\label{pro:symb-soundness}
Let $\proc{\p}{\Phi}$ be a simple process such that
$\triple{\p}{\Phi}{\varnothing} \ssymc{\tr} \triple{\p'}{\Phi'}{\Set'}$,
and $\theta \in \Sol(\Phi';\Set')$. We have that
$\proc{\p}{\Phi} \sintc{\tr\theta}
\proc{\p'\lambda}{\Phi'\lambda}$ where $\lambda$ is the first-order
solution of $\triple{\p'}{\Phi'}{\Set'}$ associated to $\theta$.
\end{proposition}

\begin{proof}
This proof can be done by induction on the length of the derivation
$\triple{\p}{\Phi}{\varnothing} \ssymc{\tr}
\triple{\p'}{\Phi'}{\Set'}$ considering two different cases when
dealing with the last block of actions: the case
of an application of the rule {\sc Block} and the case of {\sc Failure}.

To establish the soundness of the derivation made of one block of
actions, we do an induction on the proof tree witnessing this
derivation. At each step, we make use of the soundness of the standard 
symbolic semantics (as it has been shown, \eg in~\cite{CCD-tcs13}).
\qed
\end{proof}

\begin{proposition}
\label{pro:symb-completeness}
Let $\proc{\p}{\Phi}$ be a simple process such that
$\proc{\p}{\Phi} \sintc{\tr} \proc{\p'}{\Phi'}$. There exists a
symbolic process $\triple{\p_s}{\Phi_s}{\Set}$, a solution
$\theta \in \Sol(\Phi_s;\Set)$, and a sequence $\tr_s$ such
that:
\begin{itemize}
\item $\triple{\p}{\Phi}{\varnothing} \ssymc{\tr_s}
  \triple{\p_s}{\Phi_s}{\Set}$;
\item $\proc{\p'}{\Phi'} = \proc{\p_s'\lambda}{\Phi_s\lambda}$;
  and
\item $\tr = \tr_s\theta$
\end{itemize} 
where $\lambda$ is the first-order solution of
$\triple{\p_s}{\Phi_s}{\Set}$ associated to $\theta$.
\end{proposition}

\begin{proof}
The proof of this proposition is  similar to the previous one. We rely
on the completeness of the standard symbolic semantics as shown
in~\cite{CCD-tcs13}. \qed
\end{proof}

\thrmsintcssymv*

\begin{proof}
We prove the two implications separately relying on
Proposition~\ref{pro:symb-soundness} and
Proposition~\ref{pro:symb-completeness}  to go from the symbolic compressed
semantics to its concrete counterpart and vice-versa. \qed
\end{proof}

\begin{corollary}
  \label{cor:sintc-ssymc}
For any extended simple processes $A$ and $B$, we have that:
\begin{center}
 $A \eint B
\iff A \esym B$.
\end{center}
\end{corollary}

%% file: app-differentiation.tex
\section{Proofs of Section~\ref{sec:diff}}

\lemmadiffswap*
\begin{proof}
Let $\proc{\p_0}{\Phi_0}$ be an extended simple process such that
$\triple{\p_0}{\Phi_0}{\varnothing}
\ssymc{\tr} \triple{\p}{\Phi}{\Set}$ with $\tr$ 
made of  proper blocks. Let 
$\theta\in\Sol\cs{\Phi}{\Set}$ and  $\tr_c$ be a trace
such that $\tr_c =_{\I_b} \tr\theta$.
Thanks to Proposition~\ref{pro:symb-soundness}, we have
$\proc{\p_0}{\Phi_0}
\sintc{\tr\theta} \proc{\p\lambda_\theta}{\Phi\lambda_\theta}$
where $\lambda_\theta$ is the first order solution associated to $\theta$.

By definition of $\I_a$ and $\I_b$, for all blocks $b_1 = \io{c}{u}{w}, b_2 = \io{c'}{u'}{w'}$,
we have that $b_1 \I_b b_2$ if, and only if, $a_1 \I_a a_2$ for any
action $a_1$ in $b_1$ and $a_2$ in $b_2$.
Thus we have that for all concrete traces $\tr_1$, and $\tr_2$ made
of proper blocks,
$\tr_1 =_{\I_a} \tr_2
\iff \tr_1 =_{\I_b} \tr_2$.
In particular, we have that $\tr_c =_{\I_a} \tr\theta$.

By applying Proposition~\ref{pro:comp-soundness}, we have that
$\proc{\p_0}{\Phi_0}
\LRstep{\tr\theta} \proc{\p\lambda_\theta}{\Phi\lambda_\theta}$. Using
Lemma~\ref{lem:permute}, we deduce that $\proc{\p_0}{\Phi_0}
\LRstep{\tr_c} \proc{\p\lambda_\theta}{\Phi\lambda_\theta}$, and then
relying on Proposition~\ref{pro:comp-completeness} we obtain that $\proc{\p_0}{\Phi_0}
\sintc{\tr_c} \proc{\p\lambda_\theta}{\Phi\lambda_\theta}$.

Proposition~\ref{pro:symb-completeness} gives us a symbolic execution
$\triple{\p_0}{\Phi_0}{\varnothing}
\ssymc{\tr'}
\triple{\p_s}{\Phi_s}{\Set_s}$ and a substitution
$\theta'\in\Sol(\Phi_s;\Set_s)$ such that
$\p\lambda_\theta = \p_s\lambda_\theta'$
and $\tr'\theta' = \tr_c$. 
Further, it can be seen from the proof of 
Proposition~\ref{pro:symb-completeness} that the execution of $\tr'$
actually has the same structure as the one for $\tr_c$.
In particular, they make the same choices between \textsc{Then}
and \textsc{Else} rules.
Additionally, we can assume without loss of generality that
fresh variables associated to an action in $\tr'$ are
the same as those associated with the corresponding action in $\tr$.
This means that $\tr'$ is actually a permutation of $\tr$,
and we have $\p_s = \p$ and $\Phi_s = \Phi$.
Finally, since we have $\tr'\theta' =_{\I_b} \tr\theta$, it must be
that $\theta' = \theta$.
\qed
\end{proof}


\begin{proposition} \label{prop:carac}
The trace $t\in\Bio^*$ is minimal for $\och$ if, and only if,
for all factors $aub$ of $t$ such that
$a,b\in \Bio$ and $u\in \Bio^*$, if  $b \och a$  and 
$d\och b$ for all $d\in u$,
then there exists
$c\in au$ such that $\neg (c\;\I_b\; b)$.
\end{proposition}

\begin{proof}
  In order to prove the above result, which relies on our specific
  definitions of $\I_b$ and $\och$, we are going to exploit a more
  general result by
  Anisimov and Knuth~\cite{knuth-sorting},
  that holds in any partially commutative monoid:

  \smallskip{}

\emph{Let $t\in \Bio^*$, \ie  a word on $\Bio$. We have that 
    $t$ is a minimal trace if, and only if,
    for all factors $aub$ of $t$ such that
    $a,b\in \Bio$ and $u\in \Bio^*$, if
$c\;\I_b\; b$ for  any $c \in au$ then
    $ a \och b$.}
  
  \smallskip{}

  \noindent
  We are going to prove the equivalence between the two characterizations, 
  namely:
  \\[2mm]
  \begin{tabular}{cl}
    (H) & For all factors $aub$ of $t$ such that 
    $b\och a$ and $d\och b$ for all $d\in u$,
    \\ & there exists $c\in au$ such that $\neg (c\;\I_b\; b)$.
    \\
    (AK) & For all factors $aub$ of $t$ such that 
    $c\;\I_b\; b$ for all $c\in au$,
    \\ & we have $a\och b$.
  \end{tabular}

  \smallskip{}

  \noindent ($AK \Rightarrow H$)
    Consider a factor $aub$ of $t$, such that
    $b \och a$ and $d\och b$ for all $d \in u$.
    We have to show that there exists $c \in au$ such that $\neg (c\;\I_b\; b)$.
    Assume that this is not the case, \ie for all $c\in au$, $c \;\I_b\; b$.
    By (AK) this means that $a\och b$, which is a contradiction.

  \noindent  ($H \Rightarrow AK$)
    Let $aub$ be a factor of $t$ such that
    $c\;\I_b\; b$ for all $c \in au$.
    We want to show that $a\och b$.
    Note that, for any $c \in au$, the hypothesis $c\;\I_b\; b$ implies
    that $b$ and $c$ have different channels, and therefore we have
    either $c \och b$ or $b \och c$.
  Now, the contrapositive of (H) tells us that either $\lnot (b \och a)$ or
  there is some $d\in u$ such that $\lnot (d \och b)$.
    In the first case, we have $a \och b$, and we are done.
    Otherwise, let $a'$ be the rightmost such $d$ and let $a'u'b$ be
    the corresponding suffix of $aub$. We have that $b \och a'$ and
    $d' \och b$ for all $d' \in u'$. We can thus apply (H) on $a'u'b$,
    which leads to a contradiction: there is some $c \in au$ such
    that $\lnot(c \;\I_b b)$.
\qed
\end{proof}

\lemmadiffmin*
\begin{proof}
  Let $A \ssymc{\tr} \triple{\p}{\Phi}{\Set}$
  and $\theta\in\Sol\cs{\Phi}{\Set}$.
  We have the following equivalences:

\smallskip{}
\begin{tabular}{ll}
&  $\theta \models \AllDep{\tr}$\\[1mm]
$\iff$&  for all $\constrd{\vect X}{\vect w}\in\AllDep{\tr}$ there
exists
    $X^i\in\vect X$ \\
& such that $\vect w\cap \fv(X^i\theta)\neq\varnothing$\hfill (by
    definition of $\models$)\\[1mm]
$\iff$&for all factors $\io{c_k}{X_k}{w_k}\cdot\ldots\cdot
    \io{c_n}{X_n}{w_n}$ of $\tr$ such that
    $c_n\och c_k$, \\
    & and $c_j\och c_n$ for all $j \in \{k+1, \ldots,n-1\}$, there exists
    $X^i_n\in\vect X_n$ \\
&such that $\vect w_k.\dots.\vect w_{n-1}\cap
    \fv(X^i_n\theta)\neq\varnothing$ \hfill (by definition of
    $\AllDep{\_}$)\\[1mm]
$\iff$ &  for all factors
    $\io{c_k}{u_k}{w_k}\cdot\ldots\cdot\io{c_n}{u_n}{w_n}$ of $\tr\theta$ such that
    $c_n\och c_k$,\\
&  and $c_j\och c_n$ for all $j \in \{k+1, \ldots, n-1\}$, there exists 
  $u^i_n\in \vect{u_n}$ \\
& such that $\vect w_k \cdot\ldots\cdot \vect w_{n-1} \cap
    \fv(u^i_n)\neq\varnothing$ \\[1mm]
$\iff$ & for all factors
    $\io{c_k}{u_k}{w_k}\cdot\ldots\cdot \io{c_n}{u_n}{w_n}$ of $\tr\theta$ such that
    $c_n\och c_k$, \\
&  and $c_j\och c_n$ for all $j \in \{k+1,\ldots, n-1\}$, there exists
$p\in \{k,\ldots, n-1\}$ \\
& such that $\neg (\io{c_p}{u_p}{w_p} \,\I_{b}
\, \io{c_n}{u_n}{w_{n}})$ \hfill (by definition of $\I_b$).\\
&  (Note here that $\vect w_n\cap\fv{(u)}=\varnothing$ for all
$u\in\vect u_p$.) \\[1mm]
$\iff$ & $\tr\theta$ is minimal. \hfill (by Proposition~\ref{prop:carac})
\end{tabular}
\qed
\end{proof}

\theoesymesymd*
\begin{proof}
  \noindent $(\Rightarrow)$
  Assume $A \incesymd B$ and
  consider an execution $A \ssymd{\tr} \triple{\p}{\Phi}{\Set}$
  and a solution $\theta\in\Sol^2(\Phi;\Set)$.
  We are going to establish that $B$ can execute the same trace and
  yield a symbolic process $\triple{\p'}{\Phi'}{\Set'}$ such that
  $\theta \in \Sol^2(\Phi';\Set')$, and 
such that
  $\Phi\lambda_\theta \statequiv \Phi'\lambda'_\theta$ where
  $\lambda_\theta$ (resp. $\lambda_\theta'$) is the substitution
  associated to $\theta$ w.r.t. $\cs{\Phi}{\Set}$
  (resp. $\cs{\Phi'}{\Set'}$).

  Our execution directly translates to the regular symbolic semantics:
  $A \ssymc{\tr} \triple{\p}{\Phi}{\noconstrd{\Set}}$.
  Since $A \incesym B$, we also have
  $B \ssymc{\tr} \triple{\p'}{\Phi'}{\Set'_0}$
  with $\theta\in\Sol(\Phi;\Set'_0)$, and $\Phi\lambda_\theta
  \statequiv \Phi'\lambda'_\theta$.
  We thus have
  $B \ssymd{\tr} \triple{\p'}{\Phi'}{\Set'}$
  with $\Set'_0 = \noconstrd{(\Set')}$.
  It remains to show that $\theta$ satisfies the dependency constraints
  of $\AllDep{\tr}$. This is entailed by
  $\theta\in\Sol^2(\Phi;\Set)$ and $\Set = \noconstrd{\Set}\cup\AllDep{\tr}$.
  Finally, since $\Set' = \Set'_0\cup\AllDep{\tr}$, we have that
  $\theta\in\Sol^2(\Phi';\Set')$.
\medskip{}

\noindent $(\Leftarrow)$
  Assume that $A \incesymd B$ and
  let
  $A \ssymc{\tr} \triple{\p_A}{\Phi_A}{\Set_A}$
  be an execution of $A$ and
  $\theta$ a substitution such that $\theta\in\Sol(\Phi_A;\Set_A)$.
  We distinguish two cases whether the execution of $A$
  uses {\sc Failure} or not.

  (1) No {\sc Failure}.
  In that case, $\tr$ is made of proper blocks.
  Let $\tr_m^c$ be the minimum of the class of $\tr\theta$
  (\ie $\tr_m^c = \mini(\tr\theta)$). We have
  $\tr_m^c=_{\I_b}\tr\theta$ and thus, by applying Lemma~\ref{lem:diff-swap},
  we have that there is a symbolic trace $\tr_m$ such that:
  $\tr_m^c = \tr_m\theta$,
  $A\ssymc{\tr_m}\triple{\p_A}{\Phi_A}{\Set^0_A}$
  and $\theta\in\Sol(\Phi_A;\Set^0_A)$.
  We thus have
  $A\ssymd{\tr_m}\triple{\p_A}{\Phi_A}{\Set'_A}$
  where $\Set'_A = \Set^0_A\cup\AllDep{\tr_m}$.
  Since $\tr_m\theta$ is a minimum trace,
  by Lemma~\ref{lem:diff-min}, we have that
  $\theta\models \AllDep{\tr_m}$ and
  thus $\theta\in\Sol^2(\Phi_A;\Set'_A)$.
  By hypothesis, we have
  $B\ssymd{\tr_m}\triple{\p_B}{\Phi_B}{\Set'_B}$,
  $\theta\in\Sol^2(\Phi_B;\Set'_B)$ and
  $\Phi_A\lambda_{\theta}^A \estat\Phi_B\lambda_\theta^B$.
  Again, a similar execution exists in the symbolic semantics:
  $B\ssymc{\tr_m}\triple{\p_B}{\Phi_B}{\noconstrd{(\Set'_B)}}$ and
  $\theta\in\Sol(\Phi_B;\noconstrd{(\Set'_B)})$.
  We also have that $\tr_m$ is made of proper blocks
  and $\tr_m\theta =_{\I}\tr\theta$.
  By applying Lemma~\ref{lem:diff-swap}, we obtain
  $B\ssymc{\tr}\triple{\p_B}{\Phi_B}{\Set_B}$ such that
  $\theta\in\Sol(\Phi_B;\Set_B)$.
  \newcommand{\impro}{\io{\mathrm{c}}{X}{\varnothing}}
  \newcommand{\trio}{\tr_{\mathrm{io}}}\\
  (2) With {\sc Failure}. In that case, $\tr = \trio\cdot\impro$
  where $\trio$ is made of proper blocks, and the last block of
  $\tr$ is improper.
  We have the following execution (note that we necessarily have
  that $\p_A = \varnothing$): 
  \[
  A \ssymc{\trio} \triple{\p_A^2}{\Phi_A^2}{\Set_A^2}
  \ssymc{\impro} \triple{\varnothing}{\Phi_A^2}{\Set_A}.
  \]
  Let $\tr_m^c$ be the minimum trace of the class
  of $\tr\theta$ (\ie $\tr_m^c = \mini(\tr\theta)$).
  There exists $\tr_m$ such that $\tr_m^c = \tr_m\theta$.
  We can rewrite the trace as
  $\tr_m = \tr_1\cdot\impro\cdot\tr_2$ such that there is no
  block on channel $c$ in $\tr_2$.
  Let $\tr_m' = \tr_1\cdot\tr_2$. We have that
  $\tr_m'\theta = \mini(\trio\theta)$.
  Similarly to the first case
  (applying Lemma~\ref{lem:diff-swap} and \ref{lem:diff-min})
  we obtain the two following
  executions:
  \[
  A\ssymd{\tr_1}\triple{\p_A^1}{\Phi_A^1}{\Set_A^1}
  \ssymd{\tr_2}\triple{\p_A^2}{\Phi_A^2}{\Set_A'^2},
  \ \theta_{|\fvs(\Set_A'^2)}\in\Sol^2(\Phi_A^2;\Set_A'^2)
  \]
  and
  \[
  A\ssymd{\tr_1} \triple{\p_A^1}{\Phi_A^1}{\Set_A^1},
  \ \theta_{|\fvs(\Set_A^1)}\in\Sol^2(\Phi_A^1;\Set_A^1).
  \]
  Moreover $\triple{\p_A^2}{\Phi_A^2}{\Set_A^2}\ssymc{\impro}
  \triple{\varnothing}{\Phi_A^2}{\Set_A}$ and $\tr_2$ does not contain
  any block on channel $c$, thus we have:
  \[ \triple{\p_A^1}{\Phi_A^1}{\Set_A^1}\ssymd{\impro}
    \triple{\varnothing}{\Phi_A^1}{\Set_A'^1},
  \ \theta_{|\fvs(\Set_A')}\in\Sol(\Phi_A^1;\noconstrd{(\Set_A'^1)}) \]
  Since $(\tr_1.\impro)\theta$ is a minimum trace (any prefix of a minimum
  trace is a minimum as well), we obtain by applying Lemma~\ref{lem:diff-min},
  $\theta\models\AllDep{\tr_1.\impro}$.
  And thus $\theta_{|\fvs(\Set'_A)}\in\Sol^2(\Phi_A^1;\Set_A'^1)$.
  By hypothesis, we obtain the following executions:
  \[
  B\ssymd{\tr_1}\triple{\p_B^1}{\Phi_B^1}{\Set_B^1}
  \ssymd{\tr_2}\triple{\p_B^2}{\Phi_B^2}{\Set_B'^2},
  \ \theta_{|\fvs(\Set_B'^2)}\in\Sol^2(\Phi_B^2;\Set_B'^2)
  \]
  and
  \[
  B\ssymd{\tr_1} \triple{\p_B^1}{\Phi_B^1}{\Set_B^1}
  \ssymd{\impro}   \triple{\varnothing}{\Phi_B^1}{\Set'_B},
  \ \theta_{|\fvs(\Set_B')}\in\Sol^2(\Phi_B^1;\Set_B').
  \]
  There are corresponding executions in the symbolic compressed semantics.
  Since $\tr_2$ contains no block on channel $c$
  we have that
  \[
  B\ssymc{\tr_1} \triple{\p_B^1}{\Phi_B^1}{\Set_B^1}
  \ssymd{\tr_2}\triple{\p_B^2}{\Phi_B^2}{\Set_B'^2}
  \ssymc{\impro} \triple{\varnothing}{\Phi_B^2}{\Set_B'}
  \]
  and $\theta\in\Sol(\Phi_B^2;\Set'_B)$.
  Morevoer, relying on the fact that $\tr_2$ contains no action on channel $c$
  and that $\impro$ contains no outputs,
  we have $\tr_1\cdot\tr_2\cdot\impro=_{\I_b}\tr_1\cdot\impro\cdot\tr_2$
  and thus $\tr_1\cdot\tr_2\cdot\impro =_{\I_b} \tr$.
  The Lemma~\ref{lem:diff-swap}
  provides the required conclusion.
  \qed
\end{proof}

\begin{corollary}
\label{cor:diff-coincide}
For any extended simple processes $A$ and $B$, we have that:
\begin{center}
$A \eint B \iff A \esymd B$.
\end{center}
\end{corollary}

\begin{example}
  We illustrate how the second-order semantics is weaker than the first-order
  one, unless restrictions can be placed on recipes.
  We start by observing that ``non-normal'' recipes create spurious 
  dependencies. For instance, $\pi_1(\pair{u}{w})$ has a second-order
  dependency on $w$, but it is hardly relevant since the recipe is equal
  to $u$ modulo $\E$.
  This kind of problem can easily be avoided by showing that it is enough
  to consider one representative modulo $\E$.
  But more complex cases arise quickly:
  assuming that two handles $w_1$ and $w_2$ refer to the same message,
  the (normal-form) recipe
  $\adec{\aenc{u}{\pk{\pair{w}{w_1}}}}{\pk{\pair{w}{w_2}}}$ seems to depend on $w$
  while it is just a convoluted way of deriving $u$.
  Again, this problem can be avoided by integrating more complex
  observations:
  it has been shown that for standard cryptographic primitives,
  it is not needed to consider recipes that apply a destructor on top of a 
  constructor (\eg~\cite{cheval-ccs2011}).
  It is possible that, under such assumptions, the simpler
  second-order semantic becomes interesting enough in itself.
\end{example}


\theodifffirst*
\begin{proof}
We prove the two directions separately.

\noindent $(\Rightarrow)$
Let $A$ and $B$ be two simple processes 
such that $A\incesymd B$,
Let $\tr$ be a symbolic trace and $\triple{\p_A}{\Phi_A}{\Set_A}$ be a
symbolic process such that 
$A\ssymd{\tr}\triple{\p_A}{\Phi_A}{\Set_A}$, 
 and $\theta$ be a substitution such that
$\theta\in\Sol^1(\Phi_A;\Set_A)$.
We have that $\theta\in\Sol^2(\Phi_A;\Set_A)$ as well and so by hypothesis,
there exists a symbolic process $\triple{\p_B}{\Phi_B}{\Set_B}$ such
that $B\ssymd{\tr}  \triple{\p_B}{\Phi_B}{\Set_B}$ with
$\theta\in\Sol^2(\Phi_B;\Set_B)$, and $\Phi_A\lambda^A_\theta \estat
\Phi_B\lambda^B_\theta$ where $\lambda^A_\theta$
(resp. $\lambda^B_\theta$) is the first-order substitution associated
to $\theta$ w.r.t. $\cs{\Phi_A}{\Set_A}$ (resp. $\cs{\Phi_B}{\Set_B}$).

We now have to show that $\theta \in \Sol^1(\Phi_B;\Set_B)$. We have
  that $\theta \in \Sol^1(\Phi_A;\Set_A)$.
Let $\constrd{\vect X}{\vect w}\in\AllDep{\tr}$. We know that 
for each recipe $M\in\T(D_{\Set_A(X)})$ satisfying
$M(\Phi_A\lambda^A_\theta) =_\E (X\theta)(\Phi_A\lambda^A_\theta)$,
we have $\fvp(M) \cap \vect w \neq \varnothing$. We have to show that
the same holds for $B$. Let $M$ be a recipe in $\T(D_{\Set_B(X)})$ such that
$M(\Phi_B\lambda^B_\theta) =_\E (X\theta)(\Phi_B\lambda^B_\theta)$.
Since $A\incesymd B$, we have that $M\in\T(D_{\Set_A(X)})$.
We also know that  $\Phi_A\lambda^A_\theta\estat
\Phi_B\lambda^B_\theta$, and thus   $M(\Phi_A\lambda^A_\theta) =_\E (X\theta)(\Phi_A\lambda^A_\theta)$.
Since $\theta \in \Sol^1(\Phi_A;\Set_A)$, we conclude that $\fvp(M)
\cap \vect w \neq \varnothing$. This allows us to conclude that
$\theta \in \Sol^1(\Phi_B;\Set_B)$.

\bigskip{}

\noindent $(\Leftarrow)$
 Let $A$ and $B$ be two symbolic processes
 such that $A\incesymdf B$.
We prove that 
for all symbolic trace $\tr$,
if $A\ssymd{\tr}\triple{\p_A}{\Phi_A}{\Set_A}$ and
 $\theta\in\Sol^2(\Phi_A;\Set_A)$,
 then we have that
$B\ssymd{\tr}\triple{\p_B}{\Phi_B}{\Set_B}$ such that
$\Phi_A\lambda^A_\theta\estat\Phi_B\lambda^B_\theta$ and $\theta\in\Sol^2(\Phi_B;\Set_B)$.
We reason by induction on $\och$ over all symbolic traces.
Let $\tr$ ba a symbolic trace  such that
$A\ssymd{\tr}\triple{\p_A}{\Phi_A}{\Set_A}$ and
 $\theta\in\Sol^2(\Phi_A;\Set_A)$.
If $\theta\in\Sol^1(\Phi_A;\Set_A)$ then we can conclude
relying on our hypothesis $A\incesymdf B$. Otherwise,
there exists at least one dependency constraint
$\constrd{\vect X}{\vect w}\in\Set_A$ such that $\theta \models
\constrd{\vect X}{\vect w}\in\Set_A$ but that $\theta$ does not satisfy
in the first-order sense.
Thus, for each $X_i\in\vect X$, there is a recipe
$M_{X_i}\in\T(D_{\Set_A(X)})$ satisfying
$M_{X_i}(\Phi_A\lambda^A_{\theta}) =_\E (X_i\theta)(\Phi_A\lambda^A_{\theta})$
and $\fvp(M_{X_i}) \cap \vect w = \varnothing$.
We construct such recipes for each second order variable that
does not satisfy a dependency constraint in the first order.
We define $\theta'$ as follows. 
\begin{center}
For all $Y\in\fvs(\Set_A)$, 
$Y\theta' =   M_{Y}$ when it exists and 
$Y\theta' =  Y\theta$ otherwise.
\end{center}
We obviously have that $\theta'\in\Sol(\Phi_A;\noconstrd{\Set}_A)$.
Let us show that
$B\ssymc{\tr}\triple{\p_B}{\Phi_B}{\Set^0_B}$ with
$\theta'\in\Sol(\Phi_B;\Set^0_B)$.
We distinguish  two cases whether the execution of $A$ uses
{\sc Failure} or not.
\smallskip{}

\noindent{}
(No {\sc Failure})
In that case, since $\tr\theta$ is made of proper blocks,
we can directly
apply Lemma~\ref{lem:diff-min} and Lemma~\ref{lem:diff-swap}
for $A\ssymc{\tr}\triple{\p_A}{\Phi_A}{\noconstrd{\Set}_A}$ and
$\theta'$.
We thus obtain a symbolic trace $\tr'$ such that
$A\ssymd{\tr'} \triple{\p_A}{\Phi_A}{\Set'_A}$,
$\tr'\theta = \mini(\tr\theta)$ and
$\theta'\in\Sol^2(\Phi_A;\Set'_A)$.
Since we know that $\theta'\notin\Sol^2(\Phi_A;\Set_A)$,
we have $\tr\neq \tr'$ and thus $\tr'\och \tr$.

By inductive hypothesis, we have that
$B\ssymd{\tr'}\triple{\p_B}{\Phi_B}{\Set'_B}$ such that
$\Phi_A\lambda^A_{\theta'}\estat\Phi_B\lambda^B_{\theta'}$ and $\theta'\in\Sol^2(\Phi_B;\Set'_B)$.
Since, $B\ssymc{\tr'}\triple{\p_B}{\Phi_B}{\Set'_0}$ (where $\Set'_0 =
\noconstrd{\Set'_B}$) and
$\tr\theta'=_{\Ind_b}\tr'\theta'$, Lemma~\ref{lem:diff-swap} implies
that
$B\ssymc{\tr}\triple{\p_B}{\Phi_B}{\Set^0_B}$ with
$\theta'\in\Sol(\Phi_B;\Set^0_B)$.

\smallskip{}

\noindent{}
(With {\sc Failure})
  \newcommand{\impro}{\io{\mathrm{c}}{X}{\varnothing}}
  \newcommand{\trio}{\tr_{\mathrm{io}}}
  In that case, the trace is of the form
  $\tr = \trio.\impro$ where $\trio$ is made of proper blocks
  and $\vect X$ is not empty.
  We thus have the following execution (note that we necessary have that
  $\p_A = \varnothing$):
  \[
  A \ssymd{\trio} \triple{\p_A^2}{\Phi_A}{\Set_A^2}
  \ssymd{\impro} \triple{\varnothing}{\Phi_A}{\Set_A}.
  \]
  With $\theta\in\Sol^2(\Phi_A;\Set_A)$ and
  $\theta_{|\fvs(\Set_A^2)}\in\Sol^2(\Phi_A;\Set_A^2)$.
  We also have shown that
  $\theta'\in\Sol(\Phi_A;\noconstrd{\Set_A})$ and so
  $\theta'_{|\fvs(\Set_A^2)}\in\Sol(\Phi_A;\noconstrd{(\Set_A^2)})$.
  Let $\tr_m^c$ be the minimum trace of the class
  of $\tr\theta'$ (\ie $\tr_m^c = \mini(\tr\theta')$).
  There exists $\tr_m$ such that $\tr_m^c = \tr_m\theta$.
  We can rewrite the trace as follow
  $\tr_m = \tr_1\cdot\impro\cdot\tr_2$ such that there is no
  block on channel $c$ in $\tr_2$.
  Let $\tr_m' = \tr_1\cdot\tr_2$. We have that
  $\tr_m'\theta' = \mini(\trio\theta')$.
  By using the fact that $\theta'$ is not a second-order
  solution of $(\Phi_A;\Set_A)$, we can deduce that
  $\tr_m^c\neq\tr\theta'$ and so $\tr_m\och\tr$.
  Similarly, one can obtain $\tr_m'\och\trio\och \tr$.

  Similarly to the first case
  (applying Lemma~\ref{lem:diff-swap} and \ref{lem:diff-min})
  for the trace $\tr_1\cdot\tr_2$
  we obtain the following
  execution:
  \[
  A\ssymd{\tr_1}\triple{\p_A^1}{\Phi_A^1}{\Set_A^1}
  \ssymd{\tr_2}\triple{\p_A^2}{\Phi_A}{\Set_A'^2},
  \ \theta'_{|\fvs(\Set_A'^2)}\in\Sol^2(\Phi_A;\Set_A'^2)
  \]

  Moreover $\triple{\p_A^2}{\Phi_A}{\Set_A^2}\ssymd{\impro}
  \triple{\varnothing}{\Phi_A}{\Set_A}$ and $\tr_2$ does not contain
  any block on channel $c$ then
  $\triple{\p_A^1}{\Phi_A^1}{\Set_A^1}\ssymd{\impro}
  \triple{\varnothing}{\Phi_A^1}{\Set_A'^1}$.
  But the block $\impro$ is independent with all the blocks
  in $\tr_2$, so the recipes of $\theta'$ for the block $\impro$
  satisfy the constraints in $\Set_A'^1$.
  This implies $\theta'_{|\fvs(\Set'_A)}\in\Sol(\Phi_A^1;\noconstrd{(\Set_A'^1)})$.
  Since $(\tr_1.\impro)\theta'$ is a minimum trace (any prefix of a minimum
  trace is a minimum as well), we obtain by applying Lemma~\ref{lem:diff-min},
  $\theta'\models\AllDep{\tr_1.\impro}$.
  And thus $\theta'_{|\fvs(\Set'_A)}\in\Sol^2(\Phi_A^1;\Set_A'^1)$.

  By induction hypothesis, we obtain the following executions:
  \[
  B\ssymd{\tr_1}\triple{\p_B^1}{\Phi_B^1}{\Set_B^1}
  \ssymd{\tr_2}\triple{\p_B^2}{\Phi_B}{\Set_B'^2},
  \ \theta'_{|\fvs(\Set_B'^2)}\in\Sol^2(\Phi_B;\Set_B'^2)
  \]
  and
  \[
  B\ssymd{\tr_1} \triple{\p_B^1}{\Phi_B^1}{\Set_B^1}
  \ssymd{\impro}   \triple{\varnothing}{\Phi_B^1}{\Set_B'^1},
  \ \theta'_{|\fvs(\Set'_B)}\in\Sol^2(\Phi_B^1;\Set_B'^1).
  \]

  We also have that $\Phi_B^1\lambda_{\theta'}^B\estat
  \Phi_A^1\lambda_{\theta'}^A$ and
$\Phi_B\lambda_{\theta'}^B\estat
  \Phi_A\lambda_{\theta'}^A$.
  There are similar executions in symbolic semantics and
  since $\tr_2$ contains no block on channel $c$
  we have that
  \[
  B\ssymc{\tr_1} \triple{\p_B^1}{\Phi_B^1}{\Set_B^1}
  \ssymc{\tr_2}   \triple{\p_B^2}{\Phi_B}{\Set'^2_B}
  \ssymc{\impro} \triple{\varnothing}{\Phi_B}{\Set_B'}
\]
  and $\theta'\in\Sol(\Phi_B;\Set'_B)$.
  Morevoer, relying on the fact that $\tr_2$ contains no action on channel $c$
  and that $\impro$ contains no outputs,
  we have $\tr_1\cdot\tr_2\cdot\impro=_{\I_b}\tr_1\cdot\impro\cdot\tr_2$
  and thus $\tr_1\cdot\tr_2\cdot\impro =_{\I_b} \tr$.
  The Lemma~\ref{lem:diff-swap}
  provides the required execution:
\[
B\ssymc{\tr}\triple{\p_B}{\Phi_B}{\Set^0_B},\ \theta'\in\Sol(\Phi_B;\Set^0_B)
\]
\smallskip{}

Let us show that $\theta\in\Sol(\Phi_B;\Set^0_B)$.
Note that we have for all $X\in\fvs(\Set'_A)$, either
$X\theta' = X \theta$ or
$X\theta' = M_X$. Since,
$M_{X}(\Phi_A\lambda^A_{\theta}) =_\E (X\theta)(\Phi_A\lambda^A_{\theta_{n}})
=x\lambda^A_{\theta}$ where $x$ is the first-order variable
associated to $X$, for all $X\in\fvs(\Set_A)$,
we have the following equation:
\begin{equation}
(X\theta')(\Phi_A\lambda^A_{\theta}) =_\E (X\theta)(\Phi_A\lambda^A_{\theta})
\label{eq:first-A}
\end{equation}
and thus $\lambda^A_{\theta} = \lambda^A_{\theta'}$.
Since $\Phi_A\lambda^A_{\theta'}\estat\Phi_B\lambda^B_{\theta'}$, the equations~\eqref{eq:first-A} hold
for $\Phi_B\lambda_{\theta'}$ as well. More formally, for all $X\in\fvs(\Set_B^0)$,
we have the following equation:
\begin{equation}
(X\theta')(\Phi_B\lambda^B_{\theta'}) =_\E (X\theta)(\Phi_B\lambda^B_{\theta'})
\label{eq:first-B}
\end{equation}
Since the recipes of $\theta'$ satisfy the constraints
of $\Set^0_B$, so do the recipes of $\theta$.
Then $\theta\in\Sol(\Phi_B;\Set^0_B)$ and $\lambda^B_{\theta} = \lambda^B_{\theta'}$.

Further, we have that
$B\ssymd{\tr}\triple{\p_B}{\Phi_B}{\Set_B}$ where
$\Set_B = \Set^0_B\cup\AllDep{\tr}$. Since
$\theta\in\Sol(\Phi_A;\Set_A)$ and $\Set_A = \noconstrd{\Set_A}\cup\AllDep{\tr}$,
we have $\theta\models\AllDep{\tr}$.
Hence, we have that $\theta\in\Sol^2(\Phi_B;\Set_B)$.
\qed
\end{proof}

\begin{corollary}
For any extended simple processes $A$ and $B$, we have that:
\begin{center}
$A \eint B \iff A \esymdf B$.
\end{center}
\end{corollary}
\begin{proof}
  Direct consequences of
  Corollary~\ref{cor:diff-coincide} and Theorem~\ref{theo:diff-first}.
\end{proof}